\documentclass{amsart}

\usepackage{graphicx}
\usepackage{latexsym,color}
\usepackage{graphicx}
\usepackage{amssymb}
 \usepackage{amsfonts}
 \usepackage{amsmath}
\usepackage{amsthm}

\newtheorem{thm}{Theorem}[section]
 \newtheorem{dfn}[thm]{Definition}
 \newtheorem{lem}[thm]{Lemma}
 \newtheorem{rem}[thm]{Remark}
\newtheorem{cor}[thm]{Corollary}

\newtheorem{asm}[thm]{Assumption}

\newtheorem{exm}[thm]{Example}

\numberwithin{equation}{section}
\newcommand{\cA}{\mathcal{A}}

\newcommand{\cF}{\mathcal{F}}
\newcommand{\cG}{\mathcal{G}}
\newcommand{\cH}{\mathcal{H}}
\newcommand{\Hc}{\mathcal{H}}
\newcommand{\cI}{\mathcal{I}}

\newcommand{\cL}{\mathcal{L}}
\newcommand{\Lc}{\mathcal{L}}
\newcommand{\cM}{\mathcal{M}}

 \newcommand{\cP}{\mathcal{P}}

\newcommand{\cV}{\mathcal{V}}

\newcommand{\cZ}{\mathcal{Z}}

\def \D{\mathbb{D}}
\def \E{\mathbb{E}}
\def \F{\mathbb{F}}

\def \L{\mathbb{L}}
 \def \M{\mathbb{M}}
 \def \N{\mathbb{N}}
 \def\P{\mathbb{P}}
\def \Q{\mathbb{Q}}
\def \R{\mathbb{R}}

\def \Sb{\mathbb {S}}
\def \hS{\hat{\mathbb {S}}}
\def\cad{$c\acute{a}dl\acute{a}g$\ }
\def\ML{ \M_{\Lc}}

\def \hd{\hat{\D}}
\def \d{\D}

\def\reff#1{{\rm(\ref{#1})}}

\begin{document}
\title{Martingale Optimal Transport \\
in the  Skorokhod Space }

 \thanks{Research of Dolinsky is partly supported by a
 European Union Career Integration grant CIG-618235 and research of Soner
 is partly supported
 by the ETH Foundation, the Swiss Finance Institute
 and a Swiss National Foundation grant SNF 200021$\_$153555.
 The authors thank Professors Marcel Nutz, Xiaolu Tan and Nizar Touzi
 for insightful discussions and comments.
 Part of this work was completed during the visits of the authors
 to the National University of Singapore, NUS.
 The authors would like
 to thank Professors Min Dai, Steven Kou, the Department of Mathematics
 and the Institute for Mathematical Sciences of NUS
 for their hospitality.}
 \author{Yan Dolinsky \address{
 Department of Statistics, Hebrew University of Jerusalem, Israel.
 \hspace{10pt}
 {e.mail: yan.dolinsky@mail.huji.ac.il}}}
\author{H.Mete  Soner \address{
 Department of Mathematics, ETH Zurich \&
 Swiss Finance Institute. \hspace{10pt}
 {e.mail: hmsoner@ethz.ch}}\\
  ${}$\\
 Hebrew University of Jerusalem
 and  ETH Zurich}

\date{\today}
\begin{abstract}
The dual representation of the
martingale optimal transport
problem in the Skorokhod space
of multi dimensional $c\acute{a}dl\acute{a}g$ processes
is proved.
The dual is a minimisation
problem with constraints involving
stochastic integrals and is similar to the Kantorovich dual
of the standard optimal transport problem.  The constraints
are required to hold for very path in the Skorokhod space.
This problem has the financial
interpretation as the robust
hedging of path dependent European options.
\end{abstract}


\maketitle \markboth{Y.Dolinsky and H.M.Soner}
{Martingale Optimal Transport}
\renewcommand{\theequation}{\arabic{section}.\arabic{equation}}
\pagenumbering{arabic}

\section{Introduction}\label{sec:1}\setcounter{equation}{0}

Model independent approach to financial
markets provides hedges without referring
to a particular probabilistic structure.
It is also shown to be closely connected
to the classical Monge-Kantorovich
optimal transportation problem.
In this paper, we prove this connection
for quite general financial markets
that offer multi risky assets
with \cad(right continuous with left hand limits)
trajectories. This generality
is strongly motivated by the fact that investors
use several assets in their portfolios and the observed stock price processes
contain jump components \cite{AJ,AJ1}.
The main result is a Kantorovich type duality for the super-replication cost
of an exotic option $G$, which is simply a nonlinear function of the
whole stock trajectory.  It is well documented that
this duality is central to understanding the financial markets.
In particular, several other important
results including the fundamental theorem
of asset pricing follow from it.

As it is standard in these problems, following \cite{H}
we assume that a linear set of options $\cH$ is available for
static investment with a known price $\cL(h)$ for $h \in \cH$.
In addition to this static investment, the investor
can dynamically use stocks in her portfolio.
Let an admissible  predictable process
$\gamma$ represent this dynamic position in the stock
whose price process is denoted by $\Sb$
with values in the positive orthant $\R_+^d$.
An investment strategy $(h,\gamma)$ super-replicates
a exotic option if its final value at maturity $T$ dominates
$G$ in all possible cases, i.e.,
\begin{equation}
\label{e.super}
h(\Sb) + \int_0^T \gamma_u(\Sb)\ d\Sb_u \ge G(\Sb),
\quad \forall\ \Sb \in \D,
\end{equation}
where $ \D$ is the set of all
stock process $\Sb$ that are $c\acute{a}dl\acute{a}g$,
$\Sb_0=(1,\ldots,1)$ and continuous at maturity $T$.
Technical issues related to
the stochastic integral and admissible
strategies are discussed in Section \ref{sec2},
Definition \ref{d.admissible}.
 The minimal super-replicating cost is
then given by
\begin{eqnarray*}
V(G)&:=& \inf\{\ \Lc(h)\ :\  {\mbox{there exists
an admissible predictable process}}\ \gamma\\
&&\hspace{60pt}
{\mbox{so that}}\ (h,\gamma)\
 {\mbox{super-replicates}}\ G\ \}.
\end{eqnarray*}

As usual, the dual elements are martingale measures
$\Q$ that are consistent with the given option data.  Namely,
let $\M_{\Lc}$ be the set of all measures on $\D$
so that the canonical process $\Sb$ is a martingale
with the canonical filtration $\F$ and
$$
\E_\Q[ h] \le \Lc(h), \quad
h \in \Hc.
$$
We then have the following duality result,
\begin{equation}
\label{e.dual}
V(G)= \sup_{\Q\in \ML}\ \E_\Q[G].
\end{equation}
The above result is proved in Theorem \ref{t.main}
for $G$ that is uniformly
continuous in the Skorokhod topology and
satisfies a certain growth condition.  
In Theorem \ref{t.main}
we assume that 
\begin{equation}
\label{1.new}
|G(\Sb)|\leq C(1+|\Sb_T|)
\end{equation}
for some constant $C>0$.
Then, we relax this condition to \reff{e.growth}
in the last section.

In this paper,
we study two classes of pairs
$(\cH,\cL)$.  Namely, one and many marginal
cases.  In the first one,
this pair is defined through a  
a given probability measure $\mu$.
Then, we take
$\Hc$ to be the set of
all functions of the type
$g(\Sb_T)$ with $g \in \L^1(\R_+^d,\mu)$
and  set $\Lc(g)=\int g d\mu$.  The only assumption
on $\mu$ is that $\int x d\mu(x)= \Sb_0=(1,\ldots,1)$.

In the initial sections and in Theorem \ref{t.main},
 we prove the duality for the single
marginal case.  Then,  
in Section \ref{s.extend},  we both relax 
the growth assumption
on $G$ and consider the multi-marginal
problem. Namely, we fix a partition
$0<T_1<T_2<...<T_N=T$ and a probability measures
$\mu_1\preceq\mu_2\preceq...\preceq \mu_N$ on $R_+^d$,
where $\preceq$ denotes convex order on probability measures, i.e.,
$$
\mu \preceq \nu \quad
\Leftrightarrow
\quad
\int \Phi d\mu \le \int \Phi d\nu, \qquad \forall\
\Phi\ {\mbox{convex, integrable}}.
$$
We extend the duality result (\ref{e.dual}) for the case where $\Hc$ is the set of
all functions of the type
$\sum_{i=1}^N g_i(\Sb_{T_i})$ with $g_i \in \L^1(\R_+^d,\mu_i)$.
For this extension however,
we need to assume that
$\int |x|^p d\mu_N(x)<\infty$ for some $p>1$.
In particular, we assume that a power option is an element in the space of
static positions $\cH$.

Our approach, as in \cite{DS,DS1}, relies on a
discretization procedure.  We then use a classical min-max
theorem for the discrete approximation
and a classical constrained duality result of
F\"ollmer and Kramkov \cite{FK}.  The technical steps
are to prove that the approximations on both side
of the dual formula converge.  The multi-dimensionality
and the discontinuous behavior of the stock process
 introduce several technical difficulties.
 In particular, we introduce appropriate portfolio
constraints in the approximate discrete
markets.  This new feature of the discretization
is essential and enables us to control the error terms
due to the multi-dimensionality and the possible
discontinuities of the stock process.

Another technical difficulty originates from
the fact the set of martingale measures $\M_{\Lc}$
is not compact.  Therefore, passage to the limit
in the dual side requires probabilistic constructions.
In particular, we prove that the dual problem
as seen as a function of the probability measure
$\mu$ (with fixed $G$) has some continuity
properties.  This is proved in Section \ref{s.continuity},
Theorem \ref{t.continuity}.

For the multi-marginal case, it is not clear how to prove
Theorem \ref{t.continuity} via a probabilistic construction.
Instead, we use an additional idea in the discretization procedure.
This idea is based on a penalisation technique and 
requires that the linear space $\cH$ contains a power option.

The structure studied in this paper is
similar to that of \cite{H} and also of
\cite{BHR,CL,CO,CO1,DH,DOR,DS,GLT,H2,H3,H4,H5,H6,OHST}.
We refer the reader to the excellent survey of Hobson \cite{H1}
and to our previous papers \cite{DS,DS1} and to the references therein.
A related issue is the fundamental theorem of asset pricing (FTAP)
in these markets.  This problem in the robust setting
in discrete time is studied
in \cite{ABS}, \cite{BN} and \cite{DS1}.  \cite{ABS}  proves FTAP
in the model independent framework with a general $\cH$
containing a power option.  \cite{BN}
 considers a discrete time market
in which a set of probability measures $\cP$
is assumed.  The super replication is
defined by demanding \reff{e.super} not for every path $\Sb$
but $\P$ almost surely for every $\P \in \cP$ (i.e.,
$\cP$-quasi-surely). FTAP and duality
(under the assumption of no-arbitrage) is proved
 for a finite
dimensional $\cH$ but possibly with no
power option.
The notions of no-arbitrage considered in \cite{ABS}
and \cite{BN} are different.  In our earlier work \cite{DS1}
we  prove model-independent duality for a
discrete time market with proportional costs.
FTAP follows as a consequence of
the duality.  However, the form of FTAP depends on the
particular notion of no-arbitrage.
A discussion of different notions is also provided in \cite{DS1}.
In continuous time,
the desirable extension to the general quasi-sure setting
remains open with the exception of \cite{GLT}
in which a certain class $\cP$ is considered.

The paper is organised as follows. The main results are formulated in the next
section. In Section 3, Theorem \ref{t.main}
is proved.  In section 4, we prove a continuity
result for the dependence of the dual problem on
the measure $\mu$.
The final section, is devoted to extensions.
\vspace{10pt}

\noindent
{\bf{Notation.}}  We close this introduction with a list of some of the
notation used in this paper.
\begin{itemize}
\item $\R_+:= (0,\infty)$ is the set of all positive real numbers.
\item  $\N:= \{1,2,\ldots \}$ is the set of positive integers.
\item $\D$ is the set of all
$\R_+^d$ valued  $c\acute{a}dl\acute{a}g$
processes $\Sb$ that are continuous at $t=T$
 and also satisfy $\Sb_0=(1,\ldots,1)$;  Section \ref{sec2}.
 \item The similar set $\D([0,T];\R^d)$ for $\R^d$ valued processes
 is defined in Section \ref{sec2}.
 \item $\Sb$ is the canonical process
 and $\cF$ is the canonical filtration on $\D$;  Section \ref{sec2}.
 \item $\|\Sb\|= \sup\{|\Sb_t|\ : \ t \in [0,T] \}$.
 \item $\Hc$ is set of statically tradable options.  In this paper,
 it is the set of all functions of the form
 $h(\Sb)=g(\Sb_T)$, where $ g \in \L^1(\R_+^d,\mu)$
 for some probability measure $\mu$; see subsection \ref{ss.call}.

 \item $d$ is the Skorokhod metric on $\D$, see Section \ref{ss.D}.
 \item For a positive integer $n$ and $\Sb \in \D$,
 stopping times $\tau_k=\tau^{(n)}_k(\Sb)$'s and the random integer
 $M=M^{(n)}(\Sb)$
 are defined in subsection \ref{ss.app}.
  \item For a positive integer $n$ and $\Sb \in \D$,
 random times $\hat \tau_k=\hat \tau^{(n)}_k(\Sb)$'s  are defined in subsection \ref{ss.lift}
 as a function of the stopping times $\tau_k$'s.
 \item Maps $\hat \Pi : \D \to \hd$ and $\check \Pi, \Pi :\D \to \D$
 are constructed in subsection \ref{ss.lift}.
\end{itemize}
When possible we followed the convention that the notation $\hat{}\ $ is
reserved for  objects on the countable space $\hd$, such as $\hat \Sb$ is
a generic point in $\hd$ and $\hat \tau_k$'s are its jump times.
\vspace{10pt}

\section{Preliminaries and main results}
\label{sec2} \setcounter{equation}{0}

The financial market consists of
a savings account
which is normalised to unity
$B_t\equiv 1$
by discounting and of $d$ risky assets
with price process $\Sb_t \in \R_+^d$, $t\in [0,T]$, where $T<\infty$
is the maturity date.
Without loss of generality
we set the initial stock values to one, i.e.,
$\Sb_0=(1,\ldots,1)$.
We assume that each component
of the price process
is right continuous with left hand limits (i.e.,
a $c\acute{a}dl\acute{a}g$ process)
which is also continuous at maturity $t=T$.
$\D$ denotes the set of all
 $c\acute{a}dl\acute{a}g$ functions
 $$
 \Sb = (\Sb^{(1)},\ldots, \Sb^{(d)}):[0,T]\rightarrow \R_+^d,
 $$
 that are continuous at $t=T$
 and also satisfy $\Sb_0=(1,\ldots,1)$.
Then, any element of $\D$
can be a possible path for the stock price process.
This is the only assumption that we make on our financial market.

We set $\D([0,T]; \R^d)$  be the set of all \cad processes that take values in
$\R^d$ (rather than $\R_+^d$ as in the case of $\D$) that start from
$\Sb_0=(1,\ldots,1)$ and are continuous at $T$.

Consider a European path dependent option with the payoff $X=G(\Sb)$ where
$$
G:\D([0,T]; \R^d) \rightarrow\mathbb{R}.
$$
Although only  the values of $G$ on $\D$
are needed to define the problem, technically
we require $G$ to be defined on the larger space
$\D([0,T]; \R^d)$.  However, in almost all
cases extension of a function defined on $\D$ to
$\D([0,T]; \R^d)$  is straightforward; see Remark \ref{r.G}
below.

In probability theory, most
processes are required to be either
progressively measurable or predictable
with respect to a filtration.  In the context of this paper,
the natural filtration is canonical filtration
generated by the canonical process.  Then,
we have the following equivalent definition
of progressive measurability.

\begin{dfn}
\label{d.progressively measurable}
{\rm{We say that a process
$ \gamma:[0,T]\times \D \rightarrow \mathbb {R}^d $
 is}} progressively measurable {\rm {if
for any $\Sb, \tilde \Sb \in \D$ and $t\in [0,T]$,
\begin{equation} \label{e.adapted}
\Sb_u=\tilde{\Sb}_u, \ \ \forall u \in[0,t]
\ \ \Rightarrow \ \ \gamma_t(\Sb)=\gamma_t(\tilde \Sb).
\end{equation}
}} \qed
 \end{dfn}

 It is well known that if $\gamma$ is left continuous
 and progressive measurable, then it is
 predictable with respect to the canonical
 filtration.  Hence,
 in the sequel we check the predictability
 of any left continuous process by verifying \reff{e.adapted}.
\vspace{10pt}

\subsection{Tradable Options} \label{ss.call}

$\Hc$ represents the set of all options available for trading.
Although in this paper we use a specific class,
in general it is assumed to be a linear subset of real-valued functions on $\D$.
It is always assumed that
$$
h_{cash}, h_{1}, \ldots, h_{d} \in \Hc,
\quad {\mbox{where}}
\quad
h_{cash}\equiv 1,\quad
h_{k}(\Sb):= \Sb^{(k)}_T,
 \quad
\forall \ k=1,\ldots,d.
$$
The price of these options are given
through an operator
$$
\cL : \cH \to \R.
$$
The essential assumptions on $\cL$ are the convexity,
an appropriate continuity and
$$
\cL(h_{cash})=1, \quad
\cL(h_{k})=\Sb^{(k)}_0=1, \quad
\forall \ k=1,\ldots,d.
$$
The last condition implies
that the dual elements are
martingale measures.
So it might be interesting
to relax it so as
to allow for local martingale measures.

\begin{exm}
{\rm{
In this example, we discuss
the two examples $(\cH,\cL)$  studied in this paper.
\vspace{10pt}

{\bf 1.} Let $\mu$ be a probability measure and
let}}
\begin{equation}
\label{e.Hc}
\cH=\{\ h(\Sb)= g(\Sb_T)\ : \ g\in \L^1(\R_+^d,\mu)\}.
\end{equation}
{\rm{The pricing operator given through the}} probability measure $\mu$
{\rm{by,}}
\begin{equation}
\label{e.L}
\cL(g)= \int_{\R_+^d}\ g\ d\mu.
\end{equation}
{\rm{We assume that $\mu$ satisfies
\begin{equation}
\label{a.p}
\int h_{k} d\mu=\int x_kd\mu(x)=\Sb^{(k)}_0=1, \quad
\forall \ k=1,\ldots,d .
\end{equation}
The above linear pricing rule is equivalent
to assume that the distribution of $\Sb_T$ is known
and equal to $\mu$.
 \vspace{10pt}

 {\bf 2.}
Consider a partition
$0<T_1<T_2<...<T_N=T$ and a probability measures
$\mu_1\preceq\mu_2\preceq...\preceq \mu_N$ on $R_+^d$.
Assume that $\int |x|^p d\mu_N(x)<\infty$ for some $p>1$.
We also assume that $\mu_N$ is satisfying (\ref{a.p}).
Set
\begin{equation}
\label{e.Hc1}
\cH=\{\ h(\Sb)=\sum_{i=1}^N g_i(\Sb_{T_i})\ : \ g_i\in \L^1(\R_+^d,\mu_i)\}.
\end{equation}
In this case, the linear pricing rule is equivalent
to assume that for any $k$ the distribution of
$\Sb_{T_k}$ is equal to $\mu_k$.
We also assume the following analogue of 
\reff{a.p},
\begin{equation}
\label{a.p.multi}
\int x_kd\mu_i(x)=\Sb^{(k)}_0=1, \quad
\forall \ k=1,\ldots,d,  \ i=1,\ldots,N .
\end{equation}
}}\qed
\end{exm}
\vspace{10pt}

In this paper, to simplify the presentation
we mainly consider the first case. Namely, we assume that $(\cH,\cL)$ satisfy \reff{e.Hc}, \reff{e.L},
\reff{a.p}. In particular this case does not require
the existence of a power option as an element in the space
of all static positions.
In Section \ref{s.extend}, we assume the existence of a power option
and extend the duality to the  multi-marginal case \reff{e.Hc1}.

\begin{rem}
\label{r.Hinfty}
{\rm{  Again let the tradable options to be of the form
$h(\Sb)=g(\Sb_T)$.
However, now assume that $g$ is a bounded
and continuous function of $\R_+^d$.
Then, a careful analyzis our proof
shows that the duality result Theorem \ref{t.main}
holds for this problem with the same dual problem.  Hence,
the super-replication cost with this
class of tradable options is the equal to
the one with the larger class $g \in \L^1(\R_+^d,\mu)$.
See Remark \ref{r.boundedH} and
also Remark \ref{r.super-martingale} below.
\qed}}
\end{rem}

Next we discuss the importance of the power option.

\begin{rem}
\label{r.power}
{\rm{
The following example
highlights the role
of the power option assumed
in \reff{a.p} and it is
communicated
to us by Marcel Nutz.

Suppose $d=1$.
Let $h^*(\Sb)= \chi_{[0.5,\infty)}(\Sb_T)$ and $\cH$
be the three dimensional space spanned by $h^*, h_1,
h_{cash}$.  Further let $\cL$ be a
linear functional on $\cH$ with $\cL(h_{cash})=1$ and
$\cL(h^{*})=0$. For an exotic option $G$,
$V(G)$ is the super-replication cost.  Let $\tilde \cH$ be the
extended market that also includes the power option
with $\tilde V(G)$ as the corresponding super-replication cost.

In both markets, the investor can buy
the digital option $h^{*}$ with zero cost.  Clearly,
this implies  some kind or arbitrage since $h^{*} \ge 0$
and is not identically equal to zero.  However, for the market $\cH$
this arbitrage while agreeing with the notion introduced in \cite{BN},
does not agree with the one given in \cite{ABS}.
On the other hand, in $\tilde \cH$ there is arbitrage
in both senses and
the super-replication cost $\tilde V(G)= -\infty$.

In the smaller market $\cH$ it follows directly that $V(0)=0$.
But we claim that
there is no martingale measure that is consistent
with $\cL$.  Indeed, if there were a martingale measure $\Q$
satisfying
$$
\E_\Q[h^*] \le \cL(h^*)=0,
$$
then the support of the distribution $\mu$ of
$\Sb_T$ under $\Q$ must be a subset of $[0,0.5]$.
On the other hand, since $\Q$ is a martingale measure,
$\int x d\mu(x) =\Sb_0=1$.  Hence, the set
$\M_{\cL}$ is empty.  This means that
the duality \reff{e.dual} does not hold in $\cH$
while it holds in the market $\tilde \cH$
that contains the power option.
(Note that by convention
the supremum over an empty set
is defined to be minus infinity.)

Although the duality does not hold in $\cH$ with the dual
set $\M_{\cL}$, in this example it would hold
if one relaxes the dual set of measures to include
the  local martingale measures as well.
\qed

}}
\end{rem}
\vspace{10pt}

\subsection{Martingale Measures}
\label{ss.martingale measures}

Set $\Omega:=\D$ and let
$\mathcal{F}$ be the $\sigma$--algebra which is generated
by the cylindrical sets. Let $\mathbb{S}=(\mathbb S_t)_{0\leq t\leq T}$
be the canonical process given by $\mathbb S_t(\omega):=\omega_t$,
for all $\omega\in\Omega$.

A probability measure $\mathbb Q$ on the space
$(\Omega,{\mathcal F})$ is a {\em{martingale measure}},
if the canonical process $({\mathbb S_t})_{t=0}^T$ is a martingale
with respect to $\mathbb Q$ and $\mathbb S_0=(1,\ldots,1)$,
 $\mathbb Q$-a.s.

For a probability measure $\mu$ on $\R_+^d$, let $\mathbb{M}_\mu$
be the set of all martingale measures $\mathbb Q$ such that the probability distribution of
$\mathbb S_T$ under $\mathbb Q$ is equal to $\mu$. Observe that condition
$\int x_k d\mu(x) =1$ in \reff{a.p} is
equivalent to $\mathbb{M}_\mu\neq\emptyset$.
\vspace{10pt}

\subsection{Admissible portfolios} \label{ss.admissible}

Next, we describe the continuous time trading in the underlying asset $\Sb$.
 We essentially adopt the path-wise approach which was
already used in \cite{DS}.
However, the present setup is more delicate
than the one in \cite{DS}.
Indeed, due to the possible discontinuities of the integrator $\Sb$,
we require that the trading strategies are 
of bounded variation and left continuous.
Indeed, observe that
for any left continuous function $\gamma:[0,T]\rightarrow \mathbb{R}^d$
of bounded variation
and a $c\acute{a}dl\acute{a}g$ function $\Sb \in \D$,
we may use integration by parts (see Section 1.7 in \cite{P}) to define
$$
 \int_{0}^t \gamma_u d\Sb_u:= \gamma_t \cdot \Sb_t-\gamma_0\cdot \Sb_0
 -\int_{0}^t \Sb_u\cdot d\gamma_u,
 $$
where for $a,b \in \R^d$, $a\cdot b$ is the usual
scalar product. Furthermore,
the last term in the above right hand side is the Lebesgue-Stieltjes
integral and not the standard Riemann--Stieltjes integral which was used in \cite{DS}.

In particular, when $\gamma$ is also progressively
measurable (c.f., \reff{e.adapted}) then for any martingale measure
$\Q \in \M_\mu$, the stochastic integral $\int \gamma_u \Sb_u$
is well-defined and both the pathwise constructed integral
and the stochastic integral agree $\Q$ almost surely.
In the sequel, we use this equality repeatedly.

These considerations lead us to the following definition.

\begin{dfn}\label{d.admissible}
A semi-static portfolio {\rm{is a pair
$\phi:=(g,\gamma)$, where $g\in \L^1(\R_+^d,\mu)$
and $ \gamma:[0,T]\times \D \rightarrow \mathbb {R}^d $
is left continuous, progressively measurable and bounded variation where
$\gamma_t(\Sb)$ denotes the number of shares in the portfolio $\phi$
at  time $t$, before a transfer is made at this time.

A semi-static portfolio is}} admissible, {\rm{if for every $\Q\in \M_\mu$
the stochastic integral
$\int \gamma_u d\Sb_u$ is a
$\Q$ super-martingale.

An admissible semi-static portfolio is called}} super-replicating,
{\rm{if
$$
g(\Sb_T)+ \int_{0}^t \gamma_u(\Sb) d{\Sb}_u \geq  G(\Sb),
\ \ \forall{\Sb}\in \D.
$$

The (minimal)}} super-hedging cost {\rm{of $G$ is  defined by,}}
\begin{equation*}
V(G):=\inf\left\{\int g d\mu: \ \exists \gamma \
\mbox{such}  \ \mbox{that} \ \phi :=(g,\gamma) \ \mbox{is} \ \mbox{super-replicating} \ \right\}.
 \end{equation*}
 \qed
 \end{dfn}

 \begin{rem}
 \label{r.super martingale}

 {\rm{  The condition of admissibility depends
 on the measure $\mu$.  Hence the
 set of admissible controls and
 the super-replication cost
 also have this dependence.
  One may remove this dependence
  by considering continuous and
  bounded $g$'s instead of $\L^1(\R_+^d,\mu)$ functions.
  And for admissibility, instead of requiring
  that the stochastic integral $\int \gamma_u d\Sb_u$
  is a $\Q$ super-martingale for every $\Q \in \M_\mu$,
   one may impose the
  condition that this integral is uniformly
  bounded from below in $\Sb$.  A careful
  analyzis of the proof
  of Theorem \ref{t.main} reveals that
  the duality
  (under the hypothesis of Theorem \ref{t.main})
  holds with this smaller class of
  admissible portfolios
  and hence the super-replication cost is not changed.
  See Remarks \ref{r.boundedH} and \ref{r.super-martingale} below.

  In the case when $\mu$ satisfies
  \begin{equation}
  \label{2.rev}
  \int |x|^p  d\mu(x)<\infty
  \end{equation}
   with an
  exponent $p>1$,
  if there exists $C>0$ satisfying
  \begin{equation}
  \label{e.int}
  \int_0^t \gamma_u(\Sb)\ d\Sb_u \ge
  -C \left( 1 + \sup_{0\le u \le t}\ |\Sb_u|^p\right),
  \quad \forall \ t \in [0,T], \ \Sb \in \D,
  \end{equation}
  then the stochastic integral is a $\Q$ super-martingale
  for each $\Q \in \M_\mu$ due
  to Doob's inequality and \reff{2.rev}.

  In the sequel, we check the admissibility of $\gamma$
  by verifying either the above condition with  $p>1$ when
  \reff{2.rev} holds or
  again the above inequality but with $p=0$ when we only have \reff{a.p}.}}
\qed
\end{rem}

\subsection{Martingale optimal transport on the space $\D$}
\label{ss.D}

We continue by stating the duality result.
Since our approach relies on discretization, one requires
the regularity of the exotic option.  One may then relax this regularity
through analytical methods as we have done in \cite{DS1}.
Since the emphasis of this paper is the possible
discontinuity of the stock process
and multi-dimensionality,
we do not seek the most general condition on $G$. We first
prove the duality when $G$ is satisfying (\ref{1.new}) and uniformly continuous
in the Skorokhod topology.  We then relax this condition
in Section \ref{s.extend} below.  To state the condition
on $G$, recall  the Skorokhod metric on $D([0,T]; \R^d)$,

$$
d(\omega,\tilde \omega) := \inf_{\lambda \in \Lambda[0,T]}\
\sup_{t \in [0,T]}\ \left( |\omega(t)-\tilde \omega(\lambda(t))| +
|\lambda(t)-t|\right),
$$
where $\Lambda[0,T]$ is the set of all strictly increasing onto functions
$\lambda :[0,T]\to [0,T]$.
\begin{asm}
\label{a.uniformG}
We assume that the exotic option
$$
G :\D([0,T]; \R^d) \to \R,
$$
is
satisfying \reff{1.new}
and uniformly continuous, i.e.,
there exists a continuous bounded function
{\rm{(}}modulus of continuity{\rm{)}} $m_G:\R_+\to \R_+$
so that
$$
|G(\omega)-
G(\tilde\omega)|\leq m_G\left( d(\omega, \tilde\omega)\right),
\ \ \forall \omega,\tilde\omega\in
\D([0,T]; \R^d).
$$
\qed
\end{asm}
Examples (for $d=1$) of payoffs which satisfy Assumption \ref{a.uniformG}
include lookback put options with fixed strike
$$G(\Sb)=(K-\min_{0\leq t\leq T} \Sb_t)^{+}$$
and lookback call option with floating strike
$$G(\Sb)=(\mathbb S_T-\min_{0\leq t\leq T} \Sb_t)^{+}.
$$
In the last section, we relax the above assumption
to Assumption \ref{a.newG}.  This extension allows for 
more options, in particular of Asian type.

\begin{rem}
\label{r.G}
{\rm{ For technical reasons, we assume that $G$ defined
not only on $\d$ but in the larger space
$\D([0,T];\R^d)$.  However, suppose that $G$ is given only on
its natural domain $\D$ rather than the whole space
$\D([0,T]; \R^d)$.  Assume that $G$ is uniformly continuous
on $\D$.  Then, one can extend $G$ to the larger space
still satisfying the above assumption and the
main duality result
is independent of the particular extension chosen.

Indeed, a direct closure argument extends $G$
to a uniformly continuous function $\check G$
defined on $\D([0,T];[0,\infty)^d)$.  Then, we define
$\tilde G(\tilde \Sb):= \check G(\tilde \Sb')$ for every
$\tilde \Sb \in \D([0,T];\R^d)$,}} where
$\tilde \Sb '^{(i)}_t:=|\tilde \Sb^{(i)}_t|$, $i=1,...,d$ and
$t\in [0,T]$.
\qed
\end{rem}

The following result is an extension of Theorem 2.7 in \cite{DS}
to the case of multi-dimensional
stock price process with possible jumps.
Its proof is completed in the subsequent sections.
Relaxations of the Assumption \ref{a.uniformG}
are provided in Section \ref{s.extend}.

\begin{thm}
\label{t.main}
We assume that $(\Hc,\Lc)$ is as in \reff{e.Hc}, \reff{e.L}
and the probability measure $\mu$ satisfies
\reff{a.p}. Then
for for any exotic option satisfying
  Assumption \ref{a.uniformG}, we have the dual representation
  for the minimal super-replication cost defined in Definition
  \ref{d.admissible},
$$
V(G)=\sup_{\mathbb Q\in\mathbb{M}_\mu}
 \E_{\mathbb {Q}}\left[G(\mathbb
S)\right], $$
where $\E_{\mathbb {Q}}$ denotes the expectation
with respect to the probability measure $\mathbb Q$.
\end{thm}
\begin{proof}

Let $\Q \in \M_\mu$.  Then, for any
admissible strategy $\gamma$,
the path-wise integral $\int \gamma_u d\Sb_u$
agrees with stochastic integral
$\Q$-almost surely and
in view of Definition \ref{d.admissible}
this integral
is a  $\Q$ super-martingale.
Now suppose that $(g,\gamma)$
be an admissible super-replicating
semi-static portfolio.
Then,
$$
\E_\Q\left[ \int_0^T \gamma_u(\Sb)\ d\Sb_u\right]
\le 0,\quad
{\mbox{and}}
\quad
\E_\Q[g(\Sb_T)]= \int g d\mu.
$$
 We take
the expected value with respect to $\Q$
in the super-replication inequality and
use the above observations
to arrive at,
$$
V(G)\geq \sup_{\mathbb Q\in\mathbb{M}_\mu}
\E_{\mathbb {Q}}\left[G(\mathbb S)\right].
$$
The opposite inequality is proved in Corollary \ref{c.main},
\begin{equation}
\label{e.inequality}
V(G) \le
\liminf_{n\rightarrow\infty}V^{(n)}(G)\leq \sup_{\mathbb Q\in\mathbb{M}_\mu}
\E_{\mathbb {Q}}\left[G(\mathbb S)\right].
\end{equation}

We continue with a standard step that allows us to consider
only bounded and non-negative claims.
\vspace{5pt}

\noindent
{\em{Reduction to bounded non negative claims.}}
Let $C>0$ be the constant in  the assumption \reff{1.new} and set
$$
\hat{G}(\Sb):= {G}(\Sb)+C \left[1+\sum_{i=1}^d \Sb^{(i)}_T\ \right].
$$
Then, 
$$
V(\hat{G})=V(G)+(d+1)C,
$$
and also
$$
\sup_{\mathbb Q\in\mathbb{M}_\mu}
\E_{\mathbb {Q}}\left[\hat{G}(\mathbb S)\right]=
\sup_{\mathbb Q\in\mathbb{M}_\mu}
\E_{\mathbb {Q}}\left[G(\mathbb S)\right]+(d+1)C.
$$
Since by \reff{1.new} $\hat{G} \ge 0$, 
we may assume without loss of generality that 
the  claim $G$ is non negative and
satisfies Assumption \ref{a.uniformG}.

Next, for any constant $K \ge 0$ set
$G_K:= G \wedge C(dK+1) $
where $C$ is again as in (\ref{1.new}).
Then, $G_K$ is a bounded, non negative function.
Then, in view of \reff{1.new}, 
$$
G = G_{K}+ (G-C (dK+1))^{+}
\leq  G_K+C \sum_{i=1}^d (\Sb^{(i)}_T-K)^{+}.
$$
Consequently, we have the following inequalities,
\begin{eqnarray*}
V(G_K) &\leq & V(G) \ \leq \ 
 V(G_K)+ V\left(C\sum_{i=1}^d (\Sb^{(i)}_T-K)^{+}\right)\\
&= &
 V(G_K)+C \sum_{i=1}^d \int (x_i- K)^{+} d\mu(x).
\end{eqnarray*}
Also \reff{a.p} implies that
$$
\lim_{K\rightarrow\infty} \sum_{i=1}^d \int (x_i- K)^{+} d\mu(x)=0.
$$
Therefore, we conclude that
$V(G)=\lim_{k\rightarrow\infty} V(G_K)$ and so
if the inequality \reff{e.inequality} holds
for $G_K$,  then it holds for $G$.

We conclude that with without loss of generality, we may assume that
$G$ is a bounded non negative function satisfying
Assumption \ref{a.uniformG}.

\end{proof}

\begin{rem}
\label{r.super-martingale}
{\rm{  In the above proof, the lower
bound for $V(G)$
follows from a classical direct argument.
For this argument the minimal
conditions for $(g,\gamma)$
are the ones assumed in Definition
\ref{d.admissible}.  Namely,
the integrability of $g$ with respect
to $\mu$ and the super-martingality of the stochastic
integral.  Therefore,
any smaller class of semi-static portfolios
would also
satisfy the lower bound trivially.
}}
\qed
\end{rem}

\section{Proof of (\ref{e.inequality}) for bounded non negative $G$.}
\label{s.proof}
In this section and the next,
 we  assume that $G$ is bounded non negative and satisfies
Assumption \ref{a.uniformG} and that the pair $(\cH, \cL)$ is as \reff{e.Hc}, \reff{e.L}.

\subsection{Discretization of $\R_+^d$ and stopping times}
 \label{ss.app}
 In this subsection, we construct a sequence
 of stopping times that will be central to our
 discretization procedure.

 For $n \in \N$ and  $x \in \R_+^d$
 define an open set by,
  $$
  O(x,n):= \left\{ \ y \in \R_+^d\ :
 \left| y -  x\ \right| < \sqrt{d}\ 2^{-n} \right\}.
 $$
For $\Sb\in\D$, set
$\tau_0=0$ and define $\tau_{k+1}= \tau_{k+1}^{(n)}(\Sb)$ by,
$$
\tau_{k+1}:=  T \wedge
\left(\tau_k + \sqrt{d}\ 2^{-n}\right) \wedge
\inf \left\{  t>\tau_k\ :\
\Sb_ t  \not \in O(\Sb_{\tau_k},n)\
\right\}, \quad
k=0,1,\ldots \ ,
$$
where we set $\tau_{k+1}=T \wedge
(\tau_k + \sqrt{d}\ 2^{-n})$ if the above set is empty.
To ease the notation we suppress the dependence on $n$ and $\Sb$
when this dependence is clear.
Set
$$
M= \M^{(n)}(\Sb):= \min\left\{\ k\in \N \ :\ \tau_{k} = T\ \right\} .
$$
Since $\Sb$ is \cad and $\Sb \in \R_+^d$, $M<\infty$.
It is also clear that
$$
0=\tau_0<\tau_1<...<\tau_{M}=T
$$
are stopping times with respect to the filtration which is generated by $\Sb$.
Moreover, for $k=0,1,\ldots,M-1$,
\begin{equation}
\label{e.Sest}
|\tau_{k+1}- \tau_{k}|,\
\left| \Sb_t - \Sb_{\tau_k}\right| \le \sqrt{d}\  2^{-n}, \quad \forall t \in
[\tau_k, \tau_{k+1}).
\end{equation}
Also, by continuity of $\Sb$ at $T$, the above holds in the closed interval
$[\tau_{M-1}, T]$.

\subsection{Approximation}
\label{ss.countable}
In this subsection, we introduce a sequence of super-replication problems defined
on a countable probability space. In the later sections, we show that this
sequence approximates the original problem.  Since the probability space is
countable, robust (or equivalently point wise) and the probabilistic super-replications
agree with a properly chosen probability measure.  This allows us to use
classical techniques to analyze the approximating problem.

We fix $n\in\mathbb{N}$ and define a
sequence of probability spaces $\hat \D= \hat\D^{(n)}[0,T]$.
Set
\begin{eqnarray*}
 A^{(n)}&:=& \left\{ \ 2^{-n}m \ :m=(m_1,\ldots,m_d) \in \N^d \ \right\},\\
B^{(n)}&:=& \left\{k  \sqrt{d}\  2^{-n} \ :\ k \in \N \right\}
\cup \left\{ \sqrt{d}\ 2^{-n}/k \ :\ k \in \N \right\}.
\end{eqnarray*}

\begin{dfn}
\label{d.DN}
{\rm{A process $\hS \in \D$ belongs to $\hat \D$, if
there exists  a  nonnegative integer $M$
and a partition $0=t_0<t_1=\sqrt d 2^{-n}  <...<t_M<T$
such that
$$
\hS_t= \sum_{k=0}^{M-1} \hS_{t_k} \chi_{[t_{k},t_{k+1})}(t) +
\hS_{t_{M}} \chi_{[t_{M},T]}(t)
$$
where $\hS_0=(1,\dots,1)$,
$\hS_{T}=\hS_{t_M} \in A^{(n)}$ and
$$
\hS_{t_k}\in A^{(n+k)}, \  \ \forall \ k=1,\ldots,M-1,
\quad
t_k-t_{k-1}\in B^{(n+k)}, \  \ \forall \ k=2,\ldots,M.
$$
\qed}}
\end{dfn}

Since the set $\hd$ is countable, there exists a probability measure
$\mathbb{P}=\P^{(n)}$ on $\d$
with support contained in $\hd$,
which gives  positive weight to every element of $\hd$.

Let the probability structure $\Omega:=\D$,
the canonical map $\Sb$ and the filtration $\cF$ be as
in subsection \ref{ss.D}.  Introduce a new filtration
$\hat \cF=(\hat \cF_t)_{t \in [0,T]}$
by completing $\cF$ by the null sets of $\P$.  Note that
all of this structure depends on $n$ but this dependence
is suppressed in our notation.
Under the measure $\mathbb{P}$, the canonical map
${\mathbb S}$ has finitely many
jumps. Let $M=M(\Sb)$ be number of jumps and
$$
0<\hat{\tau}_1<\ldots <\hat{\tau}_{M}<T
$$
be the jump times of ${\mathbb S}$.
We set $\hat \tau_0=0, \hat \tau_{M+1}=T$.
We recall that the canonical process $\Sb$ is continuous
at $T$.

A trading strategy on the filtered probability space
$(\Omega,{\{\hat \cF_t\}}_{t=0}^T, \mathbb{P})$
is simply a predictable stochastic process $\hat\gamma$
with respect to the filtration $\hat \cF$.
Next, consider a {\em constrained} financial market, in which the
trading strategy  satisfies the bound
$$
\hat\gamma : [0,T]\times \D \to [-n\ ,\ n].
$$
The statically tradable options are bounded
real valued functions of $A^{(n)}$.

We also define a probability measure $\hat \mu$
on $A^{(n)}$  by,
$$
\hat \mu(\{ m 2^{-n}\}) := \mu\left(\left\{ x \in \R_+^d\ :\
\pi^{(n)}(x) = m 2^{-n}\right\}\right),
\quad m\in \N^d,
$$
where $\mu$ is the probability measure defining
the operator $\cL$ in subsection \ref{ss.call} and
 \begin{equation}
 \label{e.pi}
 \pi^{(n)}: \R_+^d \to A^{(n)}:= \left\{ \ 2^{-n}k \ :k=(k_1,\ldots,k_d) \in \N^d \ \right\}
 \end{equation}
is given  by
$$
 \pi^{(n)}(x)_i:=  2^{-n} \ \lceil 2^n x_i \rceil,
 \quad
 i=1,\ldots,d,
$$
and for $a \in \R_+$,  $\lceil a\rceil \in \N$ is smallest
 integer greater or equal to $a$.

We summarize this in the following by defining
the probabilistic super-replication problem on the
set $\hd$.
\begin{dfn}
 \label{d.probabilistic}
{\rm{
 A (probabilistic)}} semi-static portfolio
 {\rm{is a pair $(\hat g,\hat \gamma)$
such that
$\hat g:A^{(n)}\rightarrow \mathbb{R}$ is a bounded function,
 $ \hat \gamma :[0,T] \times \D \to[-n,n] $ {{is predictable}}
and the stochastic integral
$\int \hat\gamma_{u} d\hat{\mathbb S}_u$ exists.

A semi-static portfolio is}} admissible {\rm{if there exists $C>0$ such that
\begin{equation*}
\int_{0}^t \hat\gamma_{u} d\hat{\mathbb S}_u \ge -C ,  \ \ \mathbb P-{a.s.}, \ \ t\in [0,T].
\end{equation*}

A semi-static portfolio is}} $\mathbb P$-super-replicating, {\rm{if
\begin{equation}
\label{3.1}
\hat g( {\mathbb S}_T)+\int_{0}^T \hat \gamma_{u} d{\mathbb S}_u\geq G({\mathbb S}),
 \ \ \mathbb P-{a.s.}
\end{equation}

The (minimal) super-hedging cost of $G$ is  defined by,}}
\begin{eqnarray*}
V^{(n)}(G)&:=&\inf\left\{\int \hat  g d\hat \mu: \ \exists \gamma \
\mbox{such}  \ \mbox{that} \ \hat\phi:=(\hat g,\hat\gamma)\right.\\
&&
\hspace{60pt}
\left. \ \mbox{is admissible and super-replicating} \ \right\}.
 \end{eqnarray*}

\qed
\end{dfn}

We note that \reff{3.1} is equivalent to having the same inequality for every
$\hS \in \hd$.

\begin{rem}
\label{r.arbitrary}
{\rm{The bound $n$ that we place on the
$\gamma$ is somehow arbitrary.
Indeed, any bound that converges to infinity
with $n$ and goes to zero when multiplied by $2^{-n}$
would suffice.
This flexibility might be useful
in possible future extensions.}}
\qed
\end{rem}

\subsection{Lifting}
\label{ss.lift}

An important step
in our approach  is to ``lift''
a given  probabilistic semi-static portfolio
$\hat \phi = (\hat  h, \hat \gamma)$
to an admissible  portfolio
$\phi$ for the original financial market.

We start the construction of this lift
by defining an approximation of the the stopping times
$\tau_k=\tau_k^{(n)}(\Sb)$
defined in subsection \ref{ss.app}.
Recall also
the random integer $M=\M^{(n)}(\Sb)$
 defined in subsection \ref{ss.app}
and the set $B^{(i)}$ defined in Definition \ref{d.DN}.
Set
$$
\hat \tau_0:=0, \quad
\hat \tau_1= \sqrt{d}\ 2^{-n},\quad
\hat \tau_{M+1}:=T.
$$
For
$k=2,\ldots, M$ recursively define,
$$
\hat \tau_k:= \hat \tau_{k-1} + (1-  \sqrt{d}\ 2^{-n}/T)
\sup\left\{\ \Delta t>0 \: | \ \Delta t \in B^{(n+k)}\
{\mbox{and}}\ \Delta t< \tau_{k-1} -\tau_{k-2} \right\}.
$$
We note that due to the definition of $B^{(i)}$
the above set is always non-empty.
We collect some properties of these
random times in the following lemma.

\begin{lem}
\label{l.hattau}
Random times $\hat \tau_k$'s
satisfy,
$$
0=\hat \tau_0
< \sqrt{d}\  2^{-n} =\hat \tau_1<\ldots <\hat \tau_M<  \hat \tau_{M+1}=T,
$$
and
$$
 |\hat \tau_{k}- \tau_{k} | \le \sqrt{d}\  2^{-n+1},\quad
\forall \ k=0,\dots,M.
$$
\end{lem}
\begin{proof}

The above definitions yield,
\begin{eqnarray*}
\hat \tau_{M}&=& \hat \tau_1+\sum_{k=2}^M\ [\hat \tau_{k}-
\hat \tau_{k-1}]\\
& <&
 \sqrt{d}\ 2^{-n}+ (1- \sqrt{d}\ 2^{-n}/T)\sum_{k=2}^M\ [ \tau_{k-1}-
\tau_{k-2}] \\
&=& \sqrt{d}\ 2^{-n}+(1- \sqrt{d}\ 2^{-n}/T) [\tau_{M-1}-\tau_0]\\
&< &\sqrt{d}\ 2^{-n}+(1- \sqrt{d}\ 2^{-n}/T)T=T.
\end{eqnarray*}
This proves that
$$
0=\hat \tau_0
< \sqrt{d}\ 2^{-n}=\hat \tau_1<\ldots <\hat \tau_M<   \hat \tau_{M+1}=T.
$$
Moreover, for any $k=2,\ldots,M$,
\begin{eqnarray*}
\hat \tau_{k}&=& \hat \tau_1+\sum_{j=2}^k\ [\hat \tau_{j}-
\hat \tau_{j-1}]\\
& <&
 \sqrt{d}\ 2^{-n}+ (1- \sqrt{d}\ 2^{-n}/T)\sum_{j=2}^k\ [ \tau_{j-1}-
\tau_{j-2}] \\
&=& \sqrt{d}\ 2^{-n}+(1- \sqrt{d}\ 2^{-n}/T) [\tau_{k-1}-\tau_0]
= \tau_{k-1} +  \sqrt{d}\ 2^{-n}(1-\tau_{k-1}/T)\\
& <& \tau_{k-1} +  \sqrt{d}\ 2^{-n}.
\end{eqnarray*}

The definition of $\hat \tau_k$ and
the set $B^{(i)}$,
imply that for any $j=2,\ldots,M$,
$$
\hat \tau_{j}- \hat \tau_{j-1}
\ge \tau_{j-1}- \tau_{j-2} - \sqrt{d}\ 2^{-(n+j)}.
$$
We use this to
estimate  $\hat \tau_{k}$ with $k=2,\ldots,M$,
from below as follows.
\begin{eqnarray*}
\hat \tau_{k}&=& \hat \tau_1+\sum_{j=2}^k\ [\hat \tau_{j}-
\hat \tau_{j-1}] \\
&\ge&
 \sqrt{d}\ 2^{-n}+ (1- \sqrt{d}\ 2^{-n}/T)\sum_{j=2}^k\ [ \tau_{j-1}-
\tau_{j-2} - \sqrt{d}\ 2^{-(n+j)}] \\
&\ge& \sqrt{d}\ 2^{-n}+(1- \sqrt{d}\ 2^{-n}/T) [\tau_{k-1}-\tau_0]- \sqrt{d}\  2^{-n}
\\
&=& \tau_{k-1} - \sqrt{d}\  2^{-n}\tau_{k-1}/T\\
& >& \tau_{k-1} -  \sqrt{d}\ 2^{-n}.
\end{eqnarray*}
Since $\hat \tau_{M+1}=\tau_M=T$, $\hat \tau_1= \sqrt{d}\ 2^{-n},
\tau_0=0$, this proves that
$$
 |\hat \tau_{k}- \tau_{k-1} | \le  \sqrt{d}\ 2^{-n} ,\quad
\forall \ k=1,\dots,M+1.
$$
Also, by construction
$ |\tau_{k+1}- \tau_{k} |\le  \sqrt{d}\ 2^{-n} $
for all $k=0,\ldots,M-1$.  These inequalities complete the
proof of the
lemma.
\end{proof}

We now define a map
$\hat \Pi=\hat \Pi^{(n)}:\D\rightarrow \hd$
by,
\begin{equation}
\label{e.Pi}
\hat \Pi_{t}(\Sb):=
\sum_{k=0}^{M-1} \pi^{(n+k)}( \Sb_{\tau_k})\ \chi_{[\hat\tau_{k},\hat\tau_{k+1})}(t)
+ \pi^{(n)}(\Sb_{\tau_M})\ \chi_{[\hat\tau_{M},T]}(t),
\end{equation}
where $\pi^{(n)}$ is defined in \reff{e.pi}.

It is clear by the definition of $\pi^{(n)}$,
$\hat \tau_k$'s and Definition
\ref{d.DN}, that $\hat \Pi(\Sb)\in \hd$ for every $\Sb \in \D$.
We also note that $\Sb_{\tau_M}=S_T$
and that $\Sb$ is continuous at $T$.
For comparison, we also define
\begin{eqnarray*}
\check \Pi_{t}(\Sb)&:=&
\sum_{k=0}^{M-2} \pi^{(n+k)}(\Sb_{\tau_k})\ \chi_{[\tau_{k},\tau_{k+1})}(t)
+\pi^{(n)}( \Sb_{\tau_{M-1}})\ \chi_{[\tau_{M-1},T]}(t),\\
\Pi_{t}(\Sb)&:=&
\sum_{k=0}^{M-2} \Sb_{\tau_k}\ \chi_{[\tau_{k},\tau_{k+1})}(t)
+ \Sb_{\tau_{M-1}}\ \chi_{[\tau_{M-1},T]}(t).
\end{eqnarray*}
\begin{lem}
\label{l.pi}
Let $d$ be the Skorokhod metric.
Then, for every $\Sb \in \d$,
$$
d(\Sb,\Pi(\Sb)),\
d(\Pi(\Sb),\check \Pi(\Sb)) \le   \sqrt{d}\ 2^{-n},\ \
d(\check \Pi(\Sb), \hat \Pi(\Sb))  \le    3\sqrt{d}\ 2^{-n}.
$$
Suppose $G$ satisfies Assumption \reff{a.uniformG}.
Then,
$$
\left|G(\Sb) - G(\hat \Pi(\Sb))\right|  \le   3 m_G( 3 \sqrt{d} \ 2^{-n}).
$$
\end{lem}
\begin{proof}
In view of \reff{e.Sest}, we have,
\begin{eqnarray*}
d(\Sb,\Pi(\Sb)) &\le &
\| \Sb-\Pi(\Sb)\|_\infty\\
&=&
\max_{k=0,\ldots,M-1} \sup\{
|\Sb_t - \Sb_{\tau_k}|\ :\
t \in [\tau_k, \tau_{k+1})\ \}
\vee
|\Sb_T - \Sb_{\tau_{M-1}}|\\
&\le& \sqrt{d} 2^{-n} .
\end{eqnarray*}
Next we estimate directly that
$$
d(\Pi(\Sb),\check \Pi(\Sb)) \le
\| \Pi(\Sb)-\check \Pi(\Sb)\|_\infty
\le \sup_{x\in \R_+^d, \ k \ge 0}\ |\pi^{(n+k)}(x)-x| \le \sqrt{d} \ 2^{-n}.
$$
Define $\Lambda :[0,T] \to [0,T]$ by $\Lambda(0)=0$,
$\Lambda(\hat \tau_{k})= \tau_k$ for $k=1,\ldots,M-1$,
$$
\Lambda(\hat \tau_{M})= [\tau_{M-1}+T]/2,\quad
\Lambda(\hat \tau_{M+1})=\Lambda(T)= \tau_M=T,
$$
and to be
piecewise linear at other points.  Then, it  is clear that
$\Lambda$ is an increasing function and
$$
\check \Pi_{\Lambda(t)}(\Sb)= \hat \Pi_t(\Sb),\quad
\forall \ t \in[0,\hat \tau_{M-1}).
$$
Moreover, for $t \in [\hat \tau_{M-1},T]$,
$$
\check \Pi_{\Lambda(t)}(\Sb)= \pi^{(n)}(\Sb_{\tau_{M-1}}).
$$
Hence, by \reff{e.Sest} and the continuity of $\Sb$ at $T$,
$$
\sup_{ t \in [0,T]} \{
\ \left|  \check \Pi_{\Lambda(t)}(\Sb)
- \hat \Pi_t(\Sb) \right|\ \}=
\sup_{ t \in [\hat \tau_{M-1},T]} \{
\ \left|  \check \Pi_{\Lambda(t)}(\Sb)
- \hat \Pi_t(\Sb) \right|\ \}
\le \sqrt{d}\ 2^{-n}.
$$

We now use the above estimate together with  Lemma \ref{l.hattau}
and the above $\Lambda$ in the definition of the
Skorokhod metric.  The result is
\begin{eqnarray*}
d(\check \Pi(\Sb), \hat \Pi(\Sb)) &\le&
\sup_{t\in [0,T]}\ \{ |\check \Pi_{\Lambda(t)}(\Sb)
-  \hat \Pi_{t}(\Sb)| +
|\Lambda(t)-t|\}\\
& =&  \sqrt{d}\ 2^{-n} +
\max_{k=1,\ldots,M-1}\{
|\hat \tau_{k+1}- \tau_k|\  \}
\le  \sqrt{d}\ 2^{-n} + \sqrt{d}\ 2^{-n+1}.
\end{eqnarray*}

Suppose $G$ satisfies Assumption
\ref{a.uniformG}.
We now use the above  estimates
 to obtain
\begin{eqnarray*}
|G(\Sb)-G(\hat \Pi(\Sb))| &\le&
|G(\Sb)-G(\Pi(\Sb))| +
|G(\Pi(\Sb))-G(\check \Pi(\Sb))| \\
&&+
|G(\check \Pi(\Sb))-G(\hat \Pi(\Sb))| \\
&\le &2m_G(\sqrt{d}2^{-n}) + m_G( 3 \sqrt{d}\ 2^{-n})\\
&\le&3 m_G( 3 \sqrt{d}\ 2^{-n}).
\end{eqnarray*}

\end{proof}

We are ready to define the lift.
Let $\hat\phi=(\hat g,\hat\gamma)$
be a semi-static portfolio in the sense of
Definition \ref{d.probabilistic}. Define
a portfolio $\phi:=\Psi(\hat\phi):=(g,\gamma)$
for the original problem by
\begin{eqnarray}
\nonumber
g(x)&:=&\hat g\left(\pi^{(n)}(x)\right), \ \ x \in \R_+^d,\\
\label{e.lift}
\gamma_t(\Sb)& := &
\sum_{k=0}^{M-1} \hat\gamma_{\hat\tau_{k+1}(\Sb)}( \hat \Pi(\Sb)) \
\chi_{(\tau_k(\Sb),\tau_{k+1}(\Sb)]}(t), \ \ t\in [0,T].
\end{eqnarray}
Observe that by definition $\gamma_0(\Sb)=0$.

The following lemma provides
the important  properties of the above mapping.

\begin{lem}
\label{l.lift}
For  a semi--static portfolio
$\hat\phi=(\hat g,\hat\gamma)$
in the sense of
Definition \ref{d.probabilistic} and let
$\phi=(g,\gamma)$ be defined as in \reff{e.lift}.
Then,
$\phi$ is admissible in sense defined
in Definition \ref{d.admissible} and
has the following properties,
\begin{eqnarray*}
\int_{\R_+^d}\ g d\mu&=& \int_{A^{(n)}} \hat g d\hat \mu,\\
\left|\int_{0}^T \hat\gamma_u(\hat \Pi(\Sb))d\hat \Pi_u(\Sb)\right.
&-&\left.\int_{0}^T \gamma_u(\Sb) d\Sb_u\right|
\leq   \sqrt{d}\ n2^{-n+1},\quad
\forall \ \Sb\in \d.
\end{eqnarray*}
\end{lem}
\begin{proof}

Using the definition of $\hat \mu$ and $g$,
we directly calculate that
\begin{eqnarray*}
\int_{\R_+^d}\ g d\mu&=&  \sum_{m \in \N^d} \hat g\left(m2^{-n}\right)
\mu\left(\{ x\ :\ \pi^{(n)}(x)=m 2^{-n}\}\right)\\
&=&  \sum_{m \in \N^d} \hat g\left(m2^{-n}\right)
\hat \mu\left(\{m 2^{-n}\}\right)
=\int_{A^{(n)}} \hat g d\hat \mu.
 \end{eqnarray*}

Since $\hat \phi$ is bounded
by definition, the admissibility of $\phi$
would follow if
$\gamma$ is progressively measurable.
We show this
by verifying \reff{e.adapted}. Towards this goal,
let $\Sb, \tilde \Sb \in \D$ and $t\in [0,T]$
be such that
$\Sb_u= \tilde \Sb_u$ for all $u\leq t$.
We have to show that
$\gamma_t(\Sb) = \gamma_t(\tilde \Sb).$

Since $\gamma_0(\Sb)=\gamma_0(\tilde \Sb)=0$,
we may assume that  $t >0$.
Let $0 \le k_t(\Sb)$ be the integer such that
$t\in (\tau_{k_t}(\Sb),\tau_{k_t+1}(\Sb)].$
Since by hypothesis
$\Sb$ and $\tilde \Sb$ agree on $[0,t]$,
their jump times
up to time $t$ also agree.  In particular,
$k_t(\Sb)=k_t(\tilde \Sb)=:k_t$ and
$$
\tau_i(\Sb)=\tau_i(\tilde \Sb) <t \ \ \mbox{and} \ \
\Sb_{\tau_i(\Sb)}=\tilde \Sb_{\tau_i(\tilde \Sb)},
\quad
\forall \ i=1,\ldots, k_t.
$$
Since for any $k \ge 0$, $\hat \tau_{k+1}$ is defined
directly by $\tau_1,\ldots,\tau_k$, we also conclude that
$$
\hat \tau_i(\Sb)= \hat \tau_i(\tilde \Sb), \quad
\forall \ i =0,1,\ldots,k_t+1.
$$
Set $\theta:=\hat \tau_{k_t+1}(\Sb)=\hat \tau_{k_t+1}(\tilde \Sb)$
so that
$$
\gamma_t(\Sb)= \hat \gamma_\theta(\hat \Pi(\Sb))\quad{\mbox{and}}
\quad
\gamma_t(\tilde \Sb)= \hat \gamma_\theta(\hat \Pi(\tilde \Sb)).
$$
Since, $\hat \gamma$ is predictable, to prove
$\gamma_t(\Sb)=\gamma_t(\tilde \Sb)$
it suffices to show that
$$
\hat \Pi_u(\Sb)= \hat \Pi_u(\tilde \Sb), \quad \forall \ u < \theta.
$$

By the definition of $\hat \Pi$, for any $u<\theta$
there exists an integer $k \le k_t$ (same for both $\Sb$ and
  $\tilde \Sb$) so that
$$
\hat \Pi_u(\Sb) =  \pi(\Sb_{\tau_k}) \quad{\mbox{and}}
 \quad
  \hat \Pi_u(\tilde \Sb) =  \pi(\tilde \Sb_{\tau_k}).
  $$
 Now recall that $\Sb$ and $\tilde \Sb$ agree on
  $[0,t]$ and $\tau_k \le \tau_{k_t} <t$.  Hence,
  $ \Sb_{\tau_k}= \tilde  \Sb_{\tau_k}$ and
  consequently $ \hat \Pi_u(\Sb)=  \hat \Pi_u(\tilde \Sb)$.
  This proves that $\gamma$ is progressively measurable.

We continue by estimating the difference of the two integrals.
In view of the definitions, we have the following representations
for the stochastic integrals,
$$
\int_{0}^T \gamma_u(\Sb) d\Sb_u
=\sum_{k=1}^M \hat\gamma_{\hat\tau_{k}(\Sb)}(\hat \Pi(\Sb))
\left(\Sb_{ \tau_{k}(\Sb)}-
\Sb_{\tau_{k-1}(\Sb)}\right)
$$
and
$$
\int_{0}^T \hat\gamma_u(\hat \Pi(\Sb))d\hat \Pi_u(\Sb)
=
\sum_{k=1}^{M} \hat\gamma_{\hat\tau_{k}(\Sb)}(\hat \Pi(\Sb))
\left(\pi^{(n+k)}(\Sb_{\tau_{k}(\Sb)})-
\pi^{(n+k-1)}(\Sb_{ \tau_{k-1}(\Sb)})\right).
$$
Set
$$
\cI:=\int_{0}^T \hat\gamma_u(\hat \Pi(\Sb))d\hat \Pi_u(\Sb)
-\int_{0}^T \gamma_u(\Sb) d\Sb_u.
$$
Since the portfolio $\hat \gamma$ is bounded by ${n}$,
we have the following estimate,
$$
|\cI|  \le
2 \|\hat \gamma\|_\infty\
\sum_{k=1}^M
\left|  \pi^{(n+k)}(\Sb_{\tau_{k}(\Sb)})-
\Sb_{\tau_{k}(\Sb)}\right|
  \le   2 n \ \sum_{k=1}^M
\sqrt{d}\ 2^{-(n+k)} \le \sqrt{d}\ n2^{-n+1}.
$$
In view of the above results and the construction,
$g$ is bounded and therefore, $g \in \L^1(\R_+^d;\mu)$.
Moreover, $\gamma$ is shown to be progressively
measurable and for \\
$t\in[\tau_k,\tau_{k+1})$
\begin{eqnarray*}
\int_0^t \gamma_u(\Sb) d\Sb_u &\ge&
 \int_{0}^{\hat\tau_k} \hat\gamma_u(\hat \Pi(\Sb))d\hat \Pi_u(\Sb)
- 2 \|\hat \gamma\|_\infty
\sum_{k=1}^M
\left| \pi^{(n+k)}(\Sb_{\tau_{k}(\Sb)})-
\Sb_{\tau_{k}(\Sb)}\right|-\frac{n}{2^n}\\
&\geq& -C-\sqrt{d}\ n2^{-n+1}-\frac{n}{2^n},
\end{eqnarray*}
 where the last inequality follows
 from the fact that $\hat \gamma$ is admissible
 in the sense of Definition \ref{d.probabilistic}.
 Hence, the stochastic integral is bounded
 from below and consequently is a $\Q$ super-martingale
 for every $\Q \in \M_\mu$.  These arguments
 imply that the lifted portfolio $(g,\gamma)$ is
 admissible.
\end{proof}

The above lifting result provides an immediate
connection between $V(G)$ and $V^{(n)}(G)$.

\begin{cor}
\label{c.main}
Under the hypothesis of Theorem \ref{t.main},
 the minimal super-replication
costs satisfy
$$
V(G) \le V^{(n)}(G) + \sqrt{d} \ n2^{-n+1}+ 3 m_G(3\sqrt{d}\ 2^{-n}).
$$
In particular,
$$
V(G) \le \liminf_{n \to \infty} \ V^{(n)}(G).
$$
\end{cor}
\begin{proof}

Let $\hat \phi$
 and $\phi$ be as in Lemma \ref{l.lift}.
 Further assume that $\hat \phi$ is super-replicating
 $G$ on $\hd$.  Let $\Sb \in \d$.  Then,
 $\hat \Pi(\Sb) \in \hd$ and
 $$
\hat g (\hat \Pi_T(\Sb))
+  \int_0^T \hat \gamma_t(\hat \Pi(\Sb)) d \hat \Pi(\Sb)_t
 \ge G(\hat \Pi(\Sb)).
 $$

By  definition of $g$ and $\hat \Pi$,
$$
g(\Sb_T)=\hat g(\pi^{(n)}(\Sb_T))
= \hat g(\hat \Pi_T(\Sb)).
$$
Then, in view of  Lemma \ref{l.lift},
 \begin{eqnarray*}
  g (\Sb_T)
 +  \int_0^T \gamma_t(\Sb) d \Sb_t
  &\ge &   \hat g (\hat \Pi_T(\Sb)) + \int_0^T \hat \gamma_t(\hat \Pi(\Sb)) d \hat \Pi(\Sb)_t
-\sqrt{d}\ n2^{-n+1}\\
&\ge& G(\hat \Pi(\Sb)) -\sqrt{d}\ n2^{-n+1}\\
&\ge &
\check G(\Sb):= G(\Sb)-\sqrt{d}\ n2^{-n+1} -3m_G(3 \sqrt{d}\ 2^{-n}).
\end{eqnarray*}
Hence, $\phi$ super-replicates $\check G$.  This
implies that $\int g d\mu \ge V(\check G)$.
Since by construction $\int g d\mu = \int \hat g d\hat \mu$,
and since above inequality holds for every
super-replicating $\hat \phi$,
we conclude that
$V(\check G) \le V^{(n)}(G) $.  It is also clear that
\begin{eqnarray*}
V(G) &= &V(\check G) +\sqrt{d}\ n2^{-n+1}+ 3m_G(3 \sqrt{d}\ 2^{-n})\\
&\le&  V^{(n)}(G) + \sqrt{d}\ n2^{-n+1}+3m_G(3 \sqrt{d}\ 2^{-n}).
\end{eqnarray*}
\end{proof}

\begin{rem}
\label{r.boundedH}
{\rm{ Observe that since $\hat g$ is bounded,
so is the lifted static hedge $g$.  Hence in the
Definition \ref{d.admissible}, one may use
the class
$$
\tilde \cH := \{\ h(\Sb)=g(\Sb_T)\ :\
g \in \L^\infty(\R_+^d;\mu)\ \}.
$$
Moreover, it is not difficult to construct
$g$ so that it agrees with $\hat g$ on $A^{(n)}$
and is continuous. This construction would enable us to
consider the even smaller class $\check \cH$
with bounded and continuous $g$'s.

Moreover, in Definition \ref{d.probabilistic}
the stochastic integral $\hat \gamma_u d\Sb_u$
is assumed to be bounded from below by
a constant $C$.  In view of the above Lemma,
also the lifted portfolio satisfies that
the path wise integral $\int \gamma_u \Sb_u$
is also bounded from below, possibly
with a slightly larger constant.
This shows that in Definition \ref{d.admissible}
it would be sufficient to consider
$\gamma$'s so that the integrals
are bounded from below, instead of assuming
that their stochastic equivalents are $\Q$
super-martingales for every $ \Q \in \M_\mu$.

The above Corollary is the only
place in the proof of the upper
bound (under the hypothesis of
Theorem \ref{t.main})
where the exact definition of
admissibility is important.
Hence the above discussions
and Remark \ref{r.super-martingale}  show
that for $G$
satisfying the hypothesis of
Theorem \ref{t.main},
the super-replication cost of $G$
would be same
if we one considers the described smaller class
of admissible strategies $(g,\gamma)$.

}}
\qed
\end{rem}

\subsection{Analyzis of $V^{(n)}(G)$}
In view of the previous Corollary, to complete
the proof of \reff{e.inequality}, we need to
show the following inequality,
$$
\limsup_{n\rightarrow\infty}V^{(n)}(G)
\leq \sup_{\mathbb Q\in\mathbb{M}_\mu}
\E_{\mathbb {Q}}\left[G(\mathbb S)\right].
$$
This is done in two steps.  We first
use a standard min-max theorem and
the  constrained duality result of \cite{FK}
to get a dual representation for $V^{(n)}(G)$
(in fact we obtain an upper bound).  We then
analyze this dual by probabilistic techniques.

We start with a definition.

\begin{dfn}
\label{d.measures}
Let $\cP$ be
the set of all probability measures $\mathbb Q$
which are supported on $\hd=\D^{(n)}[0,T]$.
For $c>0$
let $\cM(n,c) \subset \cP$
be the set of all probability measures
that has the following properties,
$$
\sum_{m \in \N^d}\ \left|\Q\left(\hat\Sb_T=m 2^{-n} \right) -\hat \mu
\left(\left\{m  2^{-n} \right\}\right)\right|\leq \frac{c}{n}
$$
and
\begin{equation}
\label{e.Qn2}
\E_\Q \left[
\sum_{k=1}^{M+1}
\left|\E_\Q(\ \hat\Sb_{\hat{\tau}_k}|{\hat{\cF}}_{\hat{\tau}_k-})-
\hat \Sb_{\hat{\tau}_{k-1}}\right|\ \right]
\leq \frac{c}{n},
 \end{equation}
where as defined before,
$ \hat{\tau}_1({\hat \Sb})<\ldots
<\hat{\tau}_M({\hat \Sb})$
are
the jump times of the piecewise constant
process $\hat \Sb \in \hd$
and $\hat \tau_{0}=0, \ \hat \tau_{M+1}=T$.
\end{dfn}

We refer the reader to page 105 in \cite{P} for the definition
of the $\sigma$-algebra ${\hat{\cF}}_{\hat \tau_j-}$. Indeed, for any stopping
time $\tau \in [0,T]$, ${\hat{\cF}}_{\tau-}$ is defined to be the smallest
$\sigma$-algebra that contains ${\hat{\cF}}_0$ and all sets of the
from $A\cap \{ \tau >t\}$ for all $t \in (0,T]$ and $A \in {\hat{\cF}}_t$.
Clearly, ${\hat{\cF}}{\tau-} \subset {\hat{\cF}}_\tau$ and $\tau $ is
${\hat{\cF}}{\tau -}$ measurable.  Moreover, if $X$ is a predictable
process, then $X_\tau$ is ${\hat{\cF}}{\tau -}$ measurable (Theorem 8,
page 106 \cite{P}).

The following lemma is proved by using the results
of \cite{FK} on hedging under constraints, and applying
a classical min-max theorem.

\begin{lem}
\label{l.upperest}
Suppose that $0\le G\le c$ for some constant $c>0$.
Then,
$$
{V}^{(n)}(G) \leq \left[\sup_{\mathbb Q\in
\cM(c,n)} \E_{\mathbb Q} \left[G({\mathbb S})\right]\right]^{+},
$$
where we set  the right hand side is equals to zero
if $\cM(c,n)$ is empty.
 \end{lem}
\begin{proof}  We proceed
in several steps.

{\em{Step 1.}} In view of its definition, for any bounded
 function $\hat g$ on $A^{(n)}$, we have
$$
V^{(n)}(G)\leq \cV^{(n)}(G\ominus \hat g)+ \int g d\hat \mu,
$$
where $G\ominus \hat g(\mathbb S):=G(\Sb)-\hat g(\Sb_T)$
and for any bounded measurable real valued function $\xi$ on
$\D$,
$$
\cV^{(n)}(\xi)=\inf\left\{z\in\R:  \ \exists \gamma \ \mbox{such} \ \mbox{that} \
|\gamma|\leq n,  \  z+\int_{0}^T  \gamma_{u} d{\mathbb S}_u\geq \xi,
 \ \P-{a.s.}\right\}
$$
to be the "classical" super--hedging price of the European claim
$\xi$ under the constraint
that absolute value of the number of the stocks in the portfolio is
bounded by $n$. Furthermore, (as usual) we require that
there exists $M>0$ such that
$\int_{0}^t  \gamma_{u} d{\mathbb S}_u\geq -M$,
for every $t\in [0,T]$.

{\em{Step 2.}}
Under any measure $\mathbb Q\in \cP$ the canonical
process $\Sb$ on $\D$ is  piecewise constant with
jump times $0<\hat \tau_1<\ldots <\hat \tau_M<T$.
So it is clear that the canonical process is a $\Q$ semi-martingale.
Moreover, it has the following decomposition,
$\Sb= \M^{\Q}-A^{\Q}$
where
\begin{eqnarray}
\label{e.decomposition}
A^{\Q}_t &= & \sum_{k=1}^{M}
\chi_{[\hat\tau_k,\hat\tau_{k+1})}(t)\
\sum_{j=1}^k \ \left[\Sb_{\hat{\tau}_{j-1}}-
\E_{\Q}(\Sb_{\hat{\tau}_j}|{\hat{\cF}}_{\hat{\tau}_j-})\right], \quad \forall \ t \in [0,T)\\
\nonumber
A^{\Q}_T&:= &\lim_{t\uparrow T} A^{\Q}_t,
\end{eqnarray}
a predictable process of bounded variation and
$\M^{\mathbb Q}_t=A^{\mathbb Q}_t+\mathbb S_t$, $t\in [0,T]$, is
a $\Q$ martingale.
Then, from Example 2.3 and Proposition 4.1 in \cite{FK}
it follows that
$$
\cV^{(n)}(\xi)=\sup_{\Q\in \cP}
\E_\Q \left[\xi-
n \ \sum_{k=1}^M \left|\Sb_{\hat{\tau}_{k-1}}-
\E_Q (\Sb_{\hat{\tau}_j}|{\hat{\cF}}_{\hat{\tau}_j-})\right|
\right].
$$

{\em{Step 3.}}
Set
$$
\cZ:=\{\hat g: A^{(n)}\to \mathbb{R}: \|\hat g\|_\infty \le n \}.
$$
In view of the previous steps,
$$
V^{(n)}(G)\leq \inf_{\hat g \in\cZ} \sup_{\Q\in \cP}\ \cG(\hat g,\Q),
$$
where
$\cG:\cZ\times \cP \rightarrow\R$ is given by
$$
\cG(\hat g,\mathbb Q):=\E_\Q \left[G - n\
\sum_{k=1}^M \left|\E_\Q(\Sb_{\hat{\tau}_k}|{\hat{\cF}}_{\hat{\tau}_k-})-
\Sb_{\hat{\tau}_{k-1}}\right|\right]+
 \int \hat g d\hat \mu-\E_\Q\hat g(\Sb_T).
$$

{\em{Step 4.}}
In this step is to interchange
 the order of the infimum and supremum
 by applying a standard min-max theorem. Indeed,
consider the vector space $\R^{A^{(n)}}$
of all functions $\hat g: A^{(n)} \to \mathbb{R}$ equipped
with the topology of point-wise convergence.
Clearly,
this space is locally convex. Also, since  $A^{(n)}$
is countable, $\cZ$ is a compact subset of $\R^{A^{(n)}}$. The set
$\mathcal{P}_N$ can be
naturally considered as a convex subspace of the vector space
$\R_+^{\hd}$.
In order to apply a min-max theorem,
we also need to show continuity and concavity.

$\mathcal{G}$ is affine in the first variable,
and by the bounded convergence theorem,
it is continuous in this variable.  We claim
that $\mathcal G$ is concave in the second
variable. To this purpose, it is sufficient to show that for
any $k\geq 1$ the map
$$
\Q \to \E_\Q |\E_\Q(\Sb_{\hat{\tau}_k}
|{\hat{\cF}}_{\hat{\tau}_k-})-
\Sb_{\hat{\tau}_{k-1}}|
$$
is convex. Set
$X=\Sb_{\hat{\tau}_k}-\Sb_{\hat{\tau}_{k-1}}$,
${\hat{\cF}}:={\hat{\cF}}_{\hat{\tau}_k-}$ and $Y=\E_\Q(X|{\hat{\cF}})$.
For probability measures $\Q_1, \Q_2$ and $\lambda\in (0,1)$,
set $Y_i=\E_{\mathbb Q_i}(X|\mathcal F)$ and
$\mathbb Q=\lambda \mathbb Q_1+(1-\lambda)\mathbb Q_2$.
Then,
\begin{eqnarray*}
\E_{\mathbb Q}|Y|&=&
\E_{\mathbb Q}(Y\chi_{\{Y>0\}})-\E_{\mathbb Q}(Y\chi_{\{Y<0\}})\\
&=&
\E_\Q(X\chi_{\{Y>0\}})-\E_{\mathbb Q}(X\chi_{\{Y<0\}})\\
&=& \lambda \left(\E_{\mathbb Q_1}(X\chi_{\{Y>0\}})-\E_{\mathbb Q_1}(X\chi_{\{Y<0\}})\right)\\
&&+
(1-\lambda)\left(\E_{\mathbb Q_2}(X\chi_{\{Y>0\}})-\E_{\mathbb Q_2}(X\chi_{\{Y<0\}})\right)\\
&=& \lambda \left(\E_{\mathbb Q_1}(Y_1\chi_{\{Y>0\}})-\E_{\mathbb Q_1}(Y_1\chi_{\{Y<0\}})\right)\\
&&+
(1-\lambda)\left(\E_{\mathbb Q_2}(Y_2\chi_{\{Y>0\}})-\E_{\mathbb Q_2}(Y_2\chi_{\{Y<0\}})\right)\\
&\leq &\lambda \E_{\mathbb Q_1}|Y_1|+(1-\lambda) \E_{\mathbb Q_2}|Y_2|.
\end{eqnarray*}
This yields
the concavity of $\cG$ in the $\Q$-variable.

{\em{Step 5.}}
Next, we apply the min-max theorem, Theorem 45.8 in \cite{STR} to
$\mathcal G$.  The result is,
$$
\inf_{\hat g\in \cZ }\sup_{\Q\in \cP}\cG(\hat g,\Q)
=\sup_{\mathbb Q\in \cP }
\inf_ {\hat g\in \cZ}\cG(\hat g,\Q).
$$
Together with Step 3, we conclude that
$$
{V}^{(n)}(G)\leq
\sup_{\Q\in \cP}\inf_{h\in \cZ }
\cG(\hat g,\Q).
$$

Finally, for
any measure $\Q \in \cP$,
define $h^{\Q}\in \cZ$ by
$$
h^{\Q}\left(m2^{-n}\right)
= c \ sign \left[\Q \left(\{\Sb_T=m2^{-n}\}\right)
-\hat \mu\left(\{\Sb_T=m2^{-n}\}\right)\right],
\ \ m\in\N^d.
$$
Then, by choosing $h^\Q$ in the min-max
formula, we arrive at,
$$
V^{(n)}(G) \leq  \sup_{\Q\in\cP }
\cG(h^{\mathbb Q},\Q).
$$
Moreover,
$$
\int h^\Q d \hat \mu - \E_\Q h^\Q(\Sb_T)
= - n\ \sum_{m \in \N^d}\ \left|\Q\left(\hat\Sb_T=m 2^{-n} \right) -\hat \mu
\left(\left\{m  2^{-n} \right\}\right)\right|.
$$
Hence if $\Q$ does not belong to the
set $\cM(c,n)$, then
$$
\cG(h^{\mathbb Q},\Q) \le  \E_\Q [G(\Sb)] -c.
$$
By hypothesis, $0\le G \le c$
and therefore, $\E_\Q[G(\Sb)] \le c$ and $V^{(n)}(G) \ge 0$.
Hence, we may restrict the maximization
to $\Q \in \cM(c,n)$.
Moreover, if $\cM(c,n)$ is empty,
then we can conclude that $V^{(n)}(G) \le 0$.
\end{proof}

\subsection{Proof of \reff{e.inequality} completed.}

In order to complete the proof of Theorem
\ref{t.main}
it remains to establish the following result.

\begin{lem}
\label{l.final}
Suppose that $0 \le G\le c$ and
satisfies the  Assumption \ref{a.uniformG}. Then
\begin{equation}
\label{e.final}
 \lim\sup_{n\rightarrow\infty}
 \left[\sup_{\Q\in
\cM(c,n)} \E_{\Q} \left[G(\Sb)\right]\right]^{+}\leq
 \sup_{\Q\in\M_\mu} \E_{\Q}\left[G(\Sb)\right].
 \end{equation}

\end{lem}
\begin{proof}
Without loss of generality (by passing to a subsequence), we
may assume that the
sequence on the left hand side of \reff{e.final} is convergent.
Moreover, we may assume that
for sufficiently large $n$
the set $\cM(c,n)$ is not empty, otherwise
\reff{e.final} is trivially satisfied.

{\em{Step 1.}}
Choose $\Q_n\in \cM(c,n)$
such that
$$
\left[\sup_{\Q\in \cM(c,n)} \E_{\Q} \left[G(\Sb)\right]\right]^{+} \le 2^{-n}+
\E_{\Q_n} \left[G(\Sb)\right].
$$
Hence,
$$
\lim_{n\rightarrow\infty}\E_{\Q_n} \left[G(\Sb)\right]=
\lim\sup_{n\rightarrow\infty}
\left[\sup_{\Q\in \cM(c,n)} \E_{\Q} \left[G(\Sb)\right]\right]^{+}.
$$

Recall the decomposition
given in the second step of the proof of Lemma \ref{l.upperest}.
Set $\M^n:=\M^{\Q_n}$, $A^n:=A^{\Q_n}$.
Since $G$ is uniformly continuous in the Skorokhod metric,
$$
\left|G({\Sb})-G(\M^{n}(\Sb))\right|\leq m_G(n^{-1/2}),\quad
{\mbox{whenever}}\quad
\sup_{t\in[0,T]} A^n_t({\Sb}) \le n^{-1/2}.
$$
Therefore, since $|G(\Sb)-G(\M^{n}({\Sb}))| \le c$,
$$
\left|\E_{\Q_n}\left[G(\Sb)-G(\M^{n}({\Sb})) \right]\right|
\leq
m_G(n^{-1/2}) +
c \ {\Q_n} ( \sup_{t\in[0,T]} A^n_t \ge n^{-1/2}).
$$
We now use the representation \reff{e.decomposition} of $A^n$
together with the Markov inequality.  The result is,
$$
{\Q_n} ( \sup_{t\in[0,T]} A^n_t \geq n^{-1/2})
\le
n^{1/2}\E_{\Q_n}\sum_{k=1}^M
\left|\E_{\Q_n}(\Sb_{\hat{\tau}_k}\ |\ \mathcal{F}_{\hat{\tau}_k-})-
\Sb_{\hat{\tau}_{k-1}}\right| \le c n^{-1/2},
$$
where the last inequality follows from the fact that
$\Q_n \in \cM(c,n)$ and \reff{e.Qn2}.
Therefore, we have concluded that
$$
\lim\sup_{n\rightarrow\infty}\left[\sup_{\Q\in
\mathcal{M}(c,n)} \E_{\Q} \left[G(\Sb)\right]\right]^{+}
= \lim_{n\rightarrow\infty}\E_{\Q_n} \left[G(\M^n({\Sb}))\right].
$$

{\em{Step 2.}}
  As in subsection \ref{ss.martingale measures},
let $\tilde \Omega := \D([0,T];\R^d)$,
$\tilde \F$ be the filtration generated by the canonical
process $\tilde \Sb$.
For a probability measure $\tilde \mu$ on $\R^d$,
set $\tilde \M_{\tilde \mu}$ be set of measures $\tilde \Q$
on $\D([0,T];\R^d)$
such that the canonical process is a martingale
that starts at $\tilde \Sb_0=(1,\ldots,1)$
and the distribution of
$\tilde \Sb_T$  under $\tilde \Q$ is
equal to ${\tilde \mu}$.
Note that when the support of ${\tilde \mu}$ is on $\R_+^d$,
then the support of any measure $\tilde \Q \in \tilde \M_{\tilde \mu}$
is included in $\D$.  Hence, in that case
$\tilde \M_{\tilde \mu}$ is the same as $\M_{\tilde \mu}$ defined
earlier.

We set
\begin{equation}
\label{e.vnu}
v({\tilde \mu}):= \sup_{\tilde \Q \in \tilde \M_{\tilde \mu}}\ \E_{\tilde \Q}[G(\tilde \Sb)].
\end{equation}

Let $\tilde \Q_n$ be the measure on $\D([0,T];\R^d)$
induced by $\M^n$ under $\Q_n$, i.e.,
for any Borel subset $C \subset \D([0,T];\R^d)$,
$$
\tilde \Q_n(C):= \Q_n\left(\left\{
\Sb \in \D\ : \ \M^n(\Sb) \in C\ \right\} \right).
$$
Further, let $\nu_n$ be the distribution of $\M^n_T$ under the measure
$\Q_n$.  Since $\M^n$ is a martingale, it is clear that $\tilde \Q_n \in \tilde \M_{\nu_n}$.
Then, the previous step implies that
$$
\lim\sup_{n\rightarrow\infty}\left[\sup_{\Q\in
\mathcal{M}(c,n)} \E_{\Q} \left[G(\Sb)\right]\right]^{+}
\le \lim_{n \to \infty} v(\nu_n).
$$

{\em{Step 3.}}
Since $\Q_n \in \cM(c,n)$, \reff{e.Qn2} implies that
$$
\E_{\Q_n}\left|\Sb_T - \M^{n}_T(\Sb)\right|=
\E_{\Q_n}\left|A^{n}_T\right| \leq \frac{c}{n}.
$$
Let $\mu_n$ be the distribution
of $\Sb_T$ under $\Q_n$.  Then,
by the definition of $\cM(c,n)$,
$\mu_n$ converges weakly to $\mu$.
Then, above inequalities imply
that $\nu_n$ also converges weakly to $\mu$.

Since each component $\Sb^{(k)}_t >0$, for all
$t \in[0,T]$ and $k=1,\ldots,d$,
\begin{eqnarray*}
\E_{\Q_n}[((\M^n)^{(k)}_T(\Sb))^-] & =&
\E_{\Q_n}[(- (\M^n)^{(k)}_T(\Sb))\ \chi_{\{(\M^n)^{(k)}_T(\Sb)\le 0\}}] \\
& \le &
\E_{\Q_n}[(\Sb^{(k)}_T- (\M^n)^{(k)}_T(\Sb))\ \chi_{\{(\M^n)^{(k)}_T(\Sb)\le 0\}}] \\
& \le &
\E_{\Q_n}\left|\Sb_T - \M^{n}_T(\Sb)\right| \\
&=&
\E_{\Q_n}\left|A^{n}_T\right| \leq \frac{c}{n}.
\end{eqnarray*}
Hence, for each $k=1,\ldots,d$,
$$
\lim_{n \to \infty} \int_{\R} (x_k)^- d\nu_n(x) =0.
$$

Hence we are in a position to use the
continuity result, Theorem \ref{t.continuity} proved
in the next section.  This implies that
$$
\lim_{n \to \infty} v(\nu_n) = v(\mu).
$$
Since $\mu$ is supported on $\R_+^d$,
as remarked before, $\tilde \M_\mu = \M_\mu$.
We now combine all the steps of this proof
to arrive at,
\begin{eqnarray*}
\lim\sup_{n\rightarrow\infty}\left[\sup_{\Q\in
\mathcal{M}(c,n)} \E_{\Q} \left[G(\Sb)\right]\right]^{+}
&=& \lim_{n\rightarrow\infty}\E_{\Q_n} \left[G(\M^n({\Sb}))\right]\\
&\le & \lim_{n \to \infty} v(\nu_n) = v(\mu)\\
&=&
 \sup_{\Q\in\M_\mu} \E_{\Q}\left[G(\Sb)\right].
\end{eqnarray*}

\end{proof}

\section{Continuity of the dual with respect to $\mu$}
\label{s.continuity}
\label{sec5}\setcounter{equation}{0}

In this section, we prove  a continuity result for a martingale optimal transport problem
on the space $\D$.  Recall the functional $v({\tilde \mu})$ defined in \reff{e.vnu}
and the set of martingale measures $\tilde \M_\nu$ again defined
in \reff{e.vnu}.

\begin{thm}
\label{t.continuity}Suppose $G$ is bounded and satisfies the Assumption \ref{a.uniformG}.
Let $\nu_n$ be a sequence of probability measures
on $\R^d$.
Assume that $\nu_n$ converges weakly to a probability measure $\mu$
supported on $\R_+^d$. Further assume that for each component
$k=1,\ldots,d$,
$$
\lim_{n\rightarrow\infty}\int x^{(k)}d\nu_n(x)
=\int x^{(k)}d\mu(x)<\infty,
\quad
{\mbox{and}}
\quad
\lim_{n\rightarrow\infty}\int (x^{(k)})^{-} d\nu_n(x)=0.
$$
Then,
$$
\lim_{n \to \infty} v(\nu_n) =v(\mu).
$$
\end{thm}

\begin{proof}  To ease the notation, we take $d=1$.
First we prove that
\begin{equation}
\label{e.goal}
\lim\sup_{n\rightarrow\infty}
v(\nu_n) \leq v(\mu).
\end{equation}
In fact this the inequality that we used in the proof of Lemma \ref{l.final}.
For each $n \in \N$ choose
$\tilde Q_n\in \tilde \M_{\nu_n}$
such that
$$
 v(\nu_n) \le  2^{-n} + \E_{\Q_n}[G(\tilde \Sb)].
$$

{\em{Step 1.}}
In the first step, we construct a martingale measure
in $\M_\mu$ that is ``close'' to $\tilde Q_n$.
This construction uses the Prokhorov's metric
which we now recall.
For any two probability measures
$\nu$, $\rho$ on $\R$, the Prokhorov distance $\hat d(\nu, \rho)$
is defined to be the smallest $\delta >0$ so that
$$
\nu(C) \le \rho(C_\delta)+\delta,\quad
{\mbox{and}}
\quad
\rho(C) \le \nu(C_\delta)+\delta,
$$
for every Borel subset $C \subset \R$, where
$$
C_\delta:=\bigcup_{x\in C}(x-\delta,x+\delta).
$$
It is well
known that convergence in the Prokhorov metric is equivalent to weak convergence,
(for more details on Prokhorov's metric
we refer the reader to \cite{S}, Chapter 3, Section 7).

We now follow Theorem 4 on page 358 in \cite{S} and
Theorem 1 in \cite{Sk} to construct a random variable
$\Lambda^{(n)}$ as follows.  First construct
a probability space
$(\tilde \Omega_n,\tilde{\mathcal F}_n,\tilde {\mathbb P}_n)$
and a martingale $\M^{(n)}$ and a random variable
$\xi^{(n)}$ uniformly distributed distributed on $[0,T]$
such that:\\
a. $\xi^{(n)}$ and $\M^{(n)}$ are independent;\\
b. distribution of $\M^{(n)}$ under $\P_n$ is equal to
the measure $\tilde \Q_n$ on $\D([0,T;\R)$.  In particular,
$$
\E_{\tilde \P_n}[G(\M^{(n)})] = \E_{\Q_n}[G(\tilde \Sb)].
$$

We may choose the filtration $\tilde \cF$ to be the smallest
right-continuous filtration that is generated by the processes
$\M^{(n)}$ and $\xi^{(n)}_t:= \xi^{(n)} \wedge t$.  Recall that
$\xi^{(n)}$ is uniformly distributed on $[0,T]$ and is independent
of $\M^{(n)}$.

Moreover, in view of \cite{S,Sk} there exists a measurable
function $\psi^{(n)} : \R^2 \to \R$ such that the distribution of
$$
\Lambda^{(n)}:= \psi^{(n)}(\M^{(n)}_T, \xi^{(n)}),
$$
 on $\R$ is equal to $\mu$ and
\begin{equation}
\label{e.lambda}
\tilde {\P}_n\left(
\left|\Lambda^{(n)}-\M^{(n)}_T\right|>\hat d(\nu_n,\mu)
 \right)<\hat d(\nu_n,\mu).
\end{equation}
In particular, $\Lambda^{(n)}-\M^{(n)}_T$ converges to zero
in probability.

We set
$$
\N^{(n)}_t:= \E_{\tilde \P_n}[ \Lambda^{(n)}\ |\ \tilde \cF_t],\quad
t \in [0,T].
$$
Then, clearly $\N^{(n)}_T= \Lambda^{(n)}$
and hence has the distribution $\mu$.  Moreover,
the right-continuity of the filtration $\tilde \cF$ implies
that $\N^{(n)}$ has  has a $c\acute{a}dl\acute{a}g$
modification (for details see \cite{LS}
 Chapter 3).
 However, $\N^{(n)}_0$ is not necessarily a constant
 as $\tilde \cF_0$ may not be trivial. {\footnote{
 Authors are grateful to Professor X.~Tan of Paris, Dauphine
 for pointing out this.}}

 So we continue by modifying $\N^{(n)}$ to overcome this difficulty.
 Since $\N^{(n)}$ is right continuous than there exists $\delta>0$ such that
 for any $t\leq \delta$
 $$
 \tilde \P_n[|\N^{(n)}_t-\N^{(n)}_0|>\hat d(\nu_n,\mu)/2)]<\hat d(\nu_n,\mu)/2.
 $$
 We now define the $c\acute{a}dl\acute{a}g$ martingale
 ${\{\hat{\N}^{(n)}_t\}}_{t=0}^T$ by
 $\hat{\N}^{(n)}_t=\int xd\mu(x)=1$ for $t<\delta/2$,
 $\hat{\N}^{(n)}_t={\N}^{(n)}_{2t-\delta}$ for $\delta/2\leq t<\delta$,
   and  $\hat{\N}^{(n)}_t={\N}^{(n)}_t$ for $t\geq\delta$.
Let $\hat \cF$ be the completion of the filtration generated by
 $\hat{\N}^{(n)}$.  Then, one can directly verify that  $\hat{\N}^{(n)}$
 is a $\hat \cF$ martingale.  Therefore,
 the measure on $\D$ induced by  $\hat{\N}^{(n)}$
 under $\tilde \P_n$
 is an element in $\M_{\mu}$.  In particular,
 $$
 \E_{\tilde \P_n}[ G(\hat{\N}^{(n)})] \le  v(\mu).
 $$
In view of Assumption \ref{a.uniformG}, for any $\epsilon>0$
$$
|G(\hat{\N}^{(n)})-  G(\N^{(n)})| \le m_G(\hat d(\nu_n,\mu)+2\epsilon), \quad
{\mbox{on the set }}\quad
\cA_{n,\epsilon},
$$
where
$$
\cA_{n,\epsilon}:=
\left\{\sup_{0\leq t\leq T}\left|\hat{\N}^{(n)}_t-\N^{(n)}_t\right|>\hat d(\nu_n,\mu)+2\epsilon\right\}.
$$
Thus from the choice of delta we get
\begin{eqnarray}\label{correction}
&|\E_{\tilde \P_n}[G(\hat{\N}^{(n)})] -  \E_{\tilde \P_n}[ G(\N^{(n)})]|
\le m_G(d(\nu_n,\mu)+\epsilon) + \|G\|_\infty\ \tilde \P_n\left(\cA_{n,\epsilon}\right)\\
&\leq m_G(\hat d(\nu_n,\mu)+\epsilon)+\|G\|_\infty (\hat d(\nu_n,\mu)+
\tilde \P_n(|\N^{(n)}_0-\int x d\mu(x)|>2\epsilon))
.\nonumber
\end{eqnarray}

{\em{Step 2.}}

In view of Assumption \ref{a.uniformG},
$$
G(\M^{(n)})-  G(\N^{(n)}) \le m_G(\epsilon), \quad
{\mbox{on the set }}\quad
\cA_\epsilon^{(n)},
$$
where
$$
\cA_\epsilon^{(n)}:=
\left\{\sup_{0\leq t\leq T}\left|\M^{(n)}_t-\N^{(n)}_t\right|>\epsilon\right\}.
$$
Hence,
$$
\E_{\tilde \P_n}[G(\M^{(n)})] -  \E_{\tilde \P_n}[ G(\N^{(n)})]
\le m_G(\epsilon) + \|G\|_\infty\ \tilde \P_n\left(\cA^{(n)}_\epsilon\right).
$$

{\em{Step 3.}}
Observe that $\lim_{n\rightarrow\infty} \M^{(n)}_0=\int xd\mu(x)$ and so for sufficiently large $n$
$\{|\N^{(n)}_0-\int xd\mu(x)|>2\epsilon\}\subseteq \cA_\epsilon^{(n)}$. Thus
in view of Steps 1--2, \reff{e.goal} would follow if
$$
\lim_{n \to \infty} \tilde \P_n\left(\cA^{(n)}_\epsilon\right)=0,
$$
for each $\epsilon >0$.  Towards this goal,
we first observe that both $\M^{(n)}$ and $\N^{(n)}$ are
$(\tilde \P_n, \tilde \cF)$ martingales.  Hence,
by Doob's maximal inequality,
$$
\tilde \P_n\left(\cA^{(n)}_\epsilon\right)
\le \frac{1}{\epsilon} \E_{\tilde \P_n}\left|
\M^{(n)}_T-\N^{(n)}_T\right|.
$$
Recall that by construction $\M^{(n)}_T$ has distribution $\nu_n$
and $\N^{(n)}_T$ has distribution $\mu$.
Also by hypothesis, in the limit as $n$ tends to infinity
first moments of $\nu_n$ are
equal to those of $\mu$.  Hence,
\begin{eqnarray*}
\limsup_{n\rightarrow\infty} \E_{\tilde{\P}_n}|\N^{(n)}_T-\M^{(n)}_T|&=&
\limsup_{n\rightarrow\infty} \left[2 \E_{\tilde{\P}_n}(\N^{(n)}_T-\M^{(n)}_T)^{+}-
\E_{\tilde{\P}_n}(\N^{(n)}_T-\M^{(n)}_T)\right]\\
&=& 2 \limsup_{n\rightarrow\infty} \E_{\tilde{\P}_n}(\N^{(n)}_T-\M^{(n)}_T)^{+}.
\end{eqnarray*}

{\em{Step 4.}}  In view of \reff{e.lambda},
 $\N^{(n)}_T-\M^{(n)}_T=\Lambda^{(n)}-\M^{(n)}_T$ converges to zero
in probability.
Hence the previous step gives us the final reduction
of \reff{e.goal}.  Namely, to prove \reff{e.goal}
it suffices to show the uniform integrability of the
sequence of random variables $\M^{(n)}_T-\N^{(n)}_T$.

We first briefly recall that $\M^{(n)}_T$ has the distribution $\nu_n$,
$\N^{(n)}_T$ has the distribution $\mu$, $\mu$ is supported
on the positive real line $\R_+$ and by hypothesis
$$
\lim_{n \to \infty} \int_\R (x)^- d\nu_n(x) =-
\lim_{n \to \infty} \int_{-\infty}^0 x d\nu_n(x) =0.
$$
For brevity, set
$$
X_n:= \N^{(n)}_T,\quad
Y_n:= \M^{(n)}_T,
$$
and denote by $\E_n$ the expectation under the
measure $\tilde \P_n$.
We directly estimate that
\begin{eqnarray*}
\E_n \left[\chi_{\{X_n-Y_n>c\}}(X_n-Y_n)^{+}\right]&=&
\E_n \left[\chi_{\{X_n-Y_n>c\}}\ \chi_{\{X_n>-Y_n\}}(X_n-Y_n)^{+}\right]\\
&&+\E_n \left[\chi_{\{X_n-Y_n>c\}}\ \chi_{\{X_n<-Y_n\}}(X_n-Y_n)^{+}\right]\\
&\le &2 \E_n \left[\chi_{\{2 X_n>c\}}\ X_n\right] +
2 \E_n \left[\chi_{\{2 Y_n<-c\}}\ |Y_n|\right].
\end{eqnarray*}
Therefore,

\begin{eqnarray*}
\lim_{c\uparrow\infty}\sup_{n\in\mathbb{N}}
\E_n \left[\chi_{\{X_n-Y_n>c\}}(X_n-Y_n)^{+}\right]&\le&
2\lim_{c\uparrow\infty}\int_{c/2}^\infty xd\mu(x)\\
&&+
2 \lim_{c\uparrow\infty}\sup_{n\in\mathbb{N}}\int_{-\infty}^{-c/2} |x|d\nu_n(x)=0.
\end{eqnarray*}
This proves the uniform integrability of the sequence
 $\N^{(n)}_T-\M^{(n)}_T$.  Hence, \reff{e.goal} follows.

The opposite inequality
is proved similarly
 by replacing the roles of $\nu_n$ and $\mu$.
\end{proof}
\vspace{10pt}

\section{Extensions}
\label{s.extend}

This section discusses
the relaxations of the Assumption \ref{a.uniformG}.
Furthermore, in this section we consider the multi-marginal case.
Thus let
$0<T_1<T_2<...<T_N=T$ and
$\mu_1\preceq\mu_2\preceq...\preceq \mu_N$ be
 probability measures on $\mathbb R_+^d$
 satisfying \reff{a.p.multi}
We also assume that
$\mu_N$  satisfies
\reff{2.rev}
for some $p>1$.
The space of static positions is given by
(\ref{e.Hc1}).

In this section
we enrich the set of trading strategies, in order to deal with possible jumps at
the times $T_1,...,T_{N-1}$, ($T_N=T$ is a continuity point).
A trading strategy $\gamma={\{\gamma_t\}}_{t=0}^T$ is an admissible trading
strategy
if it has the decomposition
$\gamma=\gamma^{(1)}+\sum_{i=1}^{N-1} \beta_i\chi_{\{T_i\}}(t)$
where $\gamma^{(1)}$
is a portfolio trading strategy
 satisfies the same assumptions as
in Definition \ref{d.admissible} and $\beta_i$ is
$\mathcal{F}_{T_i-}$ measurable and bounded.
The value of such trading strategy is
given by
$$
\int_{0}^t \gamma_u d\Sb_u=\int_{0}^t \gamma^{(1)}_u d\Sb_u
+\sum_{i=1}^{N-1} (\Sb_{T_i}-\Sb_{T_i-}) \beta_i \chi_{\{T_i\leq t\}}.
$$
Thus
an admissible semi--static portfolio is a vector $(g_1,...,g_N,\gamma)$
where for any $i$, $g_i\in \mathbb{L}^{1}(\mathbb R^d_{+},\mu_i)$ and
$\gamma$ is of the above form.
An admissible semi-static portfolio is called super-replicating,
if
$$
\sum_{i=1}^N g_i(\Sb_{T_i})+ \int_{0}^t \gamma_u(\Sb) d{\Sb}_u \geq  G(\Sb),
\ \ \forall{\Sb}\in \D.
$$

The minimal super-hedging cost of $G$ is  defined by,
\begin{equation*}
V(G):=\inf\left\{\sum_{i=1}^N \int g_i d\mu_i: \ \exists \gamma \
\mbox{such}  \ \mbox{that} \ \phi :=(g_1,...,g_N,\gamma) \ \mbox{is} \ \mbox{super-replicating} \ \right\}.
 \end{equation*}
\vspace{5pt}

\begin{asm}
\label{a.newG}
We modify the Skorokhod metric and define
$$
\check{d}(\Sb, \tilde \Sb) = d(\Sb, \tilde \Sb) + \left|
 \int_0^T  \Sb_u du-  \int_0^T  \tilde \Sb_u du\right|.
 $$
 It is clear that
$$
d(\Sb, \tilde \Sb) \le \check{d}(\Sb, \tilde \Sb) \le (1+T) \|\Sb-\tilde \Sb\|.
$$
We assume that there is a modulus continuity of :  i,.e., a continuous function $m_G:[0,\infty)\to [0,\infty)$
with $m_G(0)=0$ that satisfies
$$
 \left|G(\Sb)-G(\tilde \Sb) \right| \le m_G(\check{d}(\Sb, \tilde \Sb)),\quad
 \forall \  \Sb, \tilde \Sb \in \D([0,T];\R^d).
 $$
Furthermore, we still assume that $G$ satisfies the 
following growth condition  instead of \reff{1.new},
\begin{equation}
\label{e.growth}
\left|G(\Sb)\right| \le C \left( 1+ \|\Sb\|\right),
\end{equation}
for some constant $C$.
\qed
\end{asm}

Clearly Assumption \ref{a.newG} is more general than Assumption
\ref{a.uniformG}. In particular Assumption
\ref{a.newG} allows to include
Asian call/put options with fixed and floating strikes
\begin{eqnarray*}
&G(\Sb)=\left(\frac{1}{T}\int_{0}^T \Sb_t dt-K\right)^{+}, \ \ G(\Sb)=\left(\Sb_T-\frac{1}{T}\int_{0}^T \Sb_t dt\right)^{+}, \\
&G(\Sb)=\left(K-\frac{1}{T}\int_{0}^T \Sb_t dt\right)^{+}, \ \  G(\Sb)=\left(\frac{1}{T}\int_{0}^T \Sb_t dt-\Sb_T\right)^{+},
\end{eqnarray*}
and
lookback call (respectively put) options with fixed (respectively floating) strike,
\begin{equation*}
G(\Sb)=\left(\max_{0\leq t\leq T}\Sb_t-K\right)^{+}, \ \ 
G(\Sb)=\max_{0\leq t\leq T}\Sb_t-\Sb_T.
\end{equation*}

Denote by $\mathbb{M}_{\mu_1,...,\mu_N}$ the set of all martingale measures
$\mathbb Q$ on $(\Omega,\mathcal F)$ such that
for any $k\leq N$
the probability distribution of
$\mathbb S_{T_k}$ under $\mathbb Q$ is equal to $\mu_k$. Observe that
from the relations
$\mu_1\preceq\mu_2\preceq...\preceq \mu_N$,
\reff{a.p.multi}
and the fact that $\mu_N$ is satisfying
\reff{a.p}
it follows that
$\mathbb{M}_{\mu_1,...,\mu_N}\neq\emptyset$.

The aim of this section is to prove the following result.
\begin{thm}\label{thm7.1}
Suppose that $G$ satisfies Assumption
\ref{a.newG}.  Further assume \reff{a.p.multi}
and that $\mu_N$ satisfies
\reff{2.rev}
for some $p>1$.
Then,
$$
V(G)=\sup_{\mathbb Q\in\mathbb{M}_{\mu_1,...,\mu_N}}
 \E_{\mathbb {Q}}\left[G(\mathbb
S)\right].$$
\end{thm}
\subsection{Preparation towards the proof of Theorem \ref{thm7.1}}
Towards the proof of the above theorem, we 
need several
auxiliary lemmas and modifications or previous constructions.

Set $\tau^{(1)}_0=$ and define the sequence of stopping times
$\tau^{(i)}_k$, $i=1,...,N$, $k\in\mathbb N$
by
$$
\tau^{(1)}_{1}:=
 \sqrt{d}\ 2^{-n} \wedge
\inf \left\{ t>0\ :\
\Sb_t  \not \in O(\Sb_0,n)\
\right\},
$$
and for  $k=1,\ldots, $
 $$
\tau^{(1)}_{k+1}:=  T_1 \wedge
\left(\tau^{(1)}_k + \left[\sqrt{d}\ 2^{-n} \wedge \Delta \tau^{(1)}_k\right] \right) \wedge
\inf \left\{  t>\tau^{(1)}_k\ :\
\Sb_{t}  \not \in O(\Sb_{\tau^{(1)}_k},n)\
\right\},
$$
where, $\Delta \tau^{(1)}_k=\tau^{(1)}_k-\tau^{(1)}_{k-1}$.
Set $M_1$ to be the smallest integer such that $\tau^{(1)}_{M_1}= T_1$.
Assume that we have defined
$\tau^{(i)}_k$, $i<j$, $k\in\mathbb{N}$
and $M_i$ is the smallest integer such that
$\tau^{(i)}_{M_i}= T_i$. Then, we define
$$
\tau^{(j)}_{1}:=
 T_{j-1}+\sqrt{d}\ 2^{-n} \wedge
\inf \left\{ t>T_{j-1}\ :\
\Sb_t  \not \in O(\Sb_{T_{j-1}},n)\
\right\},
$$
and for  $k=1,\ldots, $
 $$
\tau^{(j)}_{k+1}:=  T_{j-1} \wedge
\left(\tau^{(j)}_k + \left[\sqrt{d}\ 2^{-n} \wedge \Delta \tau^{(j)}_k\right] \right) \wedge
\inf \left\{  t>\tau^{(j)}_k\ :\
\Sb_{t}  \not \in O(\Sb_{\tau^{(j)}_k},n)\
\right\}.
$$
We fix $n\in\mathbb{N}$ and define a
sequence of probability spaces $\hat \D= \hat\D^{(n)}[0,T]$.
A process $\hS \in \D$ belongs to $\hat \D$, if
there exists  a nonnegative integers $M_1,...,M_N$
and a partition
\begin{eqnarray*}
0&=&t^{(1)}_0 <t^{(1)}_1=\sqrt d 2^{-n}  <...<t^{(1)}_{M_1}=\\
&&T_1=
t^{(2)}_0 <t^{(2)}_1=T_1+\sqrt d 2^{-n}  <...<t^{(2)}_{M_2}=\\
&&T_2=t^{(3)}_0
<...< t^{(N-1)}_{M_{N-1}}= \\
&&T_{N-1}
=t^{(N)}_0<t^{(2)}_1=T_{N-1}+\sqrt d 2^{-n}  <...<t^{(N)}_{M_N}<T,
\end{eqnarray*}
so that
$$
\hS_t= \sum_{i=1}^N\sum_{k=0}^{M_i-1} \hS_{t^{(i)}_k} \chi_{[t^{(i)}_{k},t^{(i)}_{k+1})}(t) +
\hS_{t^{(N)}_{M_N}} \chi_{[t^{(N)}_{M_N},T]}(t)
$$
where $\hS_0=(1,\dots,1)$, and for any
$i\leq N$ and $1\leq k< M_i$,
$$
\hS_{T_i}\in A^{(n)}, \ \ \hS_{t^{(i)}_k}\in A^{(n+k)}, \ \ t^{(i)}_{k+1}-t^{(i)}_{k}\in B^{(n+k+1)}.
$$
Once again,
the set $\hd$ is countable, thus there exists a probability measure
$\mathbb{P}=\P^{(n)}$ on $\d$
with support contained in $\hd$,
which gives  positive weight to every element of $\hd$.

The hedging problem on the countable space is given as follows.
\begin{dfn}
 \label{d.probabilisticnew}
{\rm{
 A (probabilistic)}} semi-static portfolio
 {\rm{is a pair $(\hat g_1,...,\hat g_N,\hat \gamma)$
such that for any $i$,
$\hat g_i:A^{(n)}\rightarrow \mathbb{R}$ is a bounded function and
 $ \hat \gamma :[0,T] \times \D \to[-n,n] $ is admissible trading strategy (in the same sense as in Definition \ref{d.probabilistic}).

A semi-static portfolio is}} $\mathbb P$-super-replicating, {\rm{if
$$\sum_{i=1}^N \hat g_i( {\mathbb S}_{T_i})+\int_{0}^T \hat \gamma_{u} d{\mathbb S}_u\geq G({\mathbb S}),
 \ \ \mathbb P-{a.s.}$$

The (minimal) super-hedging cost of $G$ is  defined by,}}
\begin{eqnarray*}
V^{(n)}(G)&:=&\inf\left\{\sum_{i=1}^N \int \hat  g_i d\hat \mu_i: \ \exists \gamma \
\mbox{such}  \ \mbox{that} \ \hat\phi:=(\hat g_1,...,\hat g_N,\hat\gamma)\right.\\
&&
\hspace{60pt}
\left. \ \mbox{is admissible and super-replicating} \ \right\},
 \end{eqnarray*}
where
$\hat\mu_1,...,\hat\mu_N$ are probability measures
on $A^{(n)}$  given by,
$$
\hat \mu_i(\{ m 2^{-n}\}) := \mu_i\left(\left\{ x \in \R_+^d\ :\
\pi^{(n)}(x) = m 2^{-n}\right\}\right), \ \
\quad m\in \N^d.
$$
\qed
\end{dfn}

Next, we define the lifting.
Set $T_0=0$. For any $i=1,...,N$ introduce the stopping times
$$
\hat \tau^{(i)}_0:=T_{i-1}, \quad
\hat \tau^{(i)}_1= \sqrt{d}\ 2^{-n}.
$$
For
$k=2,\ldots, M_i-1$ recursively define,
$$
\hat \tau^{(i)}_k:= \hat \tau^{(i)}_{k-1} + (1-  \sqrt{d}\ 2^{-n}/T_i)
\sup\left\{\ \Delta t>0 \: | \ \Delta t \in B^{(n+k)}\
{\mbox{and}}\ \Delta t< \tau^{(i)}_{k-1} -\tau^{(i)}_{k-2} \right\}.
$$
Also set
$\hat \tau^{(i)}_{M_i}=T_i$.

Define,
\begin{eqnarray}
\label{e.Pinew}
&\hat \Pi_{t}(\Sb):=
\sum_{i=1}^N\sum_{k=0}^{M_i-1} \pi^{(n+k)}( \Sb_{\tau^{(i)}_k})\ \chi_{[\hat\tau^{(i)}_{k},\hat\tau^{(i)}_{k+1})}(t)
+\pi^{(n)}(\Sb_T)\chi_{T}(t),\\
&\check \Pi_{t}(\Sb):=
\sum_{i=1}^N\sum_{k=0}^{M_i-1} \pi^{(n+k)}( \Sb_{\tau^{(i)}_k})\ \chi_{[\tau^{(i)}_{k},\tau^{(i)}_{k+1})}(t)
+ \pi^{(n)}(\Sb_T)\chi_{T}(t),\nonumber\\
&\Pi_{t}(\Sb):=
\sum_{i=1}^N\sum_{k=0}^{M_i-1}  \Sb_{\tau^{(i)}_k}\ \chi_{[\tau^{(i)}_{k},\tau^{(i)}_{k+1})}(t)
+\Sb_T\chi_{T}(t).\nonumber
\end{eqnarray}

Similarly to Lemma $\ref{l.pi}$
we get that
\begin{equation}\label{7.5}
d(\Sb,\Pi(\Sb)),\
d(\Pi(\Sb),\check \Pi(\Sb)) \le   \sqrt{d}\ 2^{-n},\ \
d(\check \Pi(\Sb), \hat \Pi(\Sb))  \le    3N\sqrt{d}\ 2^{-n}.
\end{equation}
The first two inequalities are proved in the same way as in Lemma \ref{l.pi}. The third inequality
done in a similar way as in Lemma \ref{l.pi} by modifying the map $\Lambda:[0,T]\rightarrow [0,T]$
as follows.
Define
$\Lambda(\hat \tau^{(i)}_{k})= \tau^{(i)}_k$ for $i=1,...,N$, $k=0,\ldots,M_i-1$,
and to be
piecewise linear at other points.

Now we estimate
$|\int_0^T  \Sb_u du-
 \int_{0}^T\hat \Pi_u(\Sb) du|.$
Fix $i<N$.
Clearly,
\begin{eqnarray*}
\left|\int_{T_{i-1}}^{T_i}  \Sb_u du-
 \int_{T_{i-1}}^{T_i}\hat \Pi_u(\Sb) du\right |&\leq & 2\sqrt d 2^{-n} \Delta T_i \|\Sb\| +
\left|\int_{T_{i-1}}^{T_i}  \check\Pi({\Sb})_u du-
 \int_{T_{i-1}}^{T_i}\hat \Pi_u(\Sb) du \right|  \\
 & \leq & 2\sqrt d 2^{-n}\Delta T_i \|\Sb\|+
 \|\Sb\| \left[ (T-\tau^{(i)}_{M_{i-1}}) + (T-\hat \tau^{(i)}_{M_{i-1}})\right]\\
 &&+\sum_{k=0}^{M_i-2} \left|\pi^{(n+k)}(\Sb_{\tau^{(i)}_k})\right|
 \left| \Delta \tau^{(i)}_{k+1} - \Delta \hat \tau^{(i)}_{k+1}\right|.
\end{eqnarray*}
Observe that for any $k=2,...,M_i$,
$$
\Delta \hat \tau^{(i)}_k \le (1-\sqrt{d}\  2^{-n}/T) \Delta\hat \tau^{(i)}_{k-1},\quad
\Delta \hat \tau^{(i)}_1=\sqrt{d}\  2^{-n}.
$$
and
\begin{eqnarray*}
\left|\pi^{(n+k)}(\Sb_{\tau^{(i)}_k})\right|  &\le& \|\Sb\| +\sqrt{d}\ 2^{-n},\\
T_i-\tau^{(i)}_{M-1} =  \Delta \tau^{(i)}_M &\le &  \Delta \tau^{(i)}_1  \sqrt{d}\ 2^{-n},\\
T_i-\hat \tau^{(i)}_{M_{i-1}} & \leq &
 \Delta \tau^{(i)}_M+\sqrt{d}2^{-n}/{T_i}\le \sqrt{d}\ 2^{-n}(1+1/T_i).
\end{eqnarray*}
Hence,
\begin{eqnarray*}
&& \left|\int_{T_{i-1}}^{T_i}  \Sb_u du-
 \int_{T_{i-1}}^{T_i}\hat \Pi_u(\Sb) du\right|
  \leq  \hat c_1 2^{-n}\|\Sb\|+
\sum_{k=0}^{M_i-2} \left|\pi^{(n+k)}(\Sb_{\tau^{(i)}_k})\right|
 \left| \Delta \tau^{(i)}_{k+1} - \Delta \hat \tau^{(i)}_{k+1}\right| \\
 &&\hspace{50pt}\le 
  [ \|\Sb\|+ \sqrt{d}\ 2^{-n}]
 \left| \Delta \tau^{(i)}_{1} - \sqrt{d}\  2^{-n}\right| 
+\hat c_1 2^{-n} \|\Sb\| \\
 &&\hspace{55pt} + \left[ \|\Sb\|+ \sqrt{d}\ 2^{-n}\right]   \sum_{k=1}^{M_i-2}
 \left| \Delta \tau^{(i)}_{k+1} - (1-\sqrt{d}\  2^{-n}/T_i) \Delta \tau^{(i)}_{k}\right|\\
 &&\hspace{50pt}\le 
\hat c_2  2^{-n} \|\Sb\|   + \  \|\Sb\|  \sum_{k=1}^{M_i-2}
 \left| \Delta \tau^{(i)}_{k+1} - \Delta \tau^{(i)}_{k}\right|
 + \|\Sb\| (\sqrt{d}\  2^{-n}/T_i)
  \sum_{k=1}^{M_i-2}
  \Delta \tau^{(i)}_{k} \\
 &&\hspace{50pt}\le 
\hat c_2  2^{-n} \|\Sb\|  +  \|\Sb\|
 [ \Delta \tau^{(i)}_{M} - \Delta \tau^{(i)}_{1}]
 + \|\Sb\|  \sqrt{d}\  2^{-n}\\
 &&\hspace{50pt}\le 
\hat c_3  2^{-n} \|\Sb\|  ,
 \end{eqnarray*}
 where $\hat c_1, \hat c_2,\hat c_3$ are appropriate constants (independent
 of $n$ and $\Sb$).
Hence,
\begin{equation}\label{7.8}
\left|\int_0^T  \Sb_u du-
 \int_{0}^T\hat \Pi_u(\Sb) du\right|\leq c_1 \|\Sb\| 2^{-n}
 \end{equation}
 for some constant $c_1$.

Finally,
let $\hat\phi=(\hat g_1,...,\hat g_N,\hat\gamma)$
be a semi-static portfolio in the sense of
Definition \ref{d.probabilistic}. Define
a portfolio $\phi:=\Psi(\hat\phi):=(g_1,...,g_N,\gamma)$
for the original problem by
\begin{eqnarray*}
g_i(x)&:=&\hat g_i\left(\pi^{(n)}(x)\right), \ \ i=1,...,N \  \  x \in \R_+^d,\\
\gamma_t(\Sb)&:=&
\sum_{i=1}^N\sum_{k=0}^{M_i-1} \hat\gamma_{\hat\tau^{(i)}_{k+1}(\Sb)}( \hat \Pi(\Sb)) \
\chi_{(\tau^{(i)}_k(\Sb),\tau^{(i)}_{k+1}(\Sb)]}(t)+
\sum_{i=1}^{N-1} (\hat\gamma_{T_i}-\hat\gamma_{\hat\tau^{(i)}_{M_i-1}}) \chi_{\{T_i\}}(t).
\nonumber
\end{eqnarray*}
As in Lemma \ref{l.lift}, 
we have that for any $i$,
\begin{equation}\label{7.8+}
\int_{\R_+^d}\ g_i d\mu_i= \int_{A^{(n)}} \hat g_i d\hat \mu_i.
\end{equation}
Furthermore,
$$
\int_{T_{i-1}}^{T_i} \gamma_u(\Sb) d\Sb_u
=\sum_{k=1}^{M_i-1} \hat\gamma_{\hat\tau^{(i)}_{k}(\Sb)}(\hat \Pi(\Sb))
\left(\Sb_{ \tau^{(i)}_{k}(\Sb)}-
\Sb_{\tau^{(i)}_{k-1}(\Sb)}\right)
+(\hat\gamma_{T_i}-\hat\gamma_{\hat\tau^{(i)}_{M_i-1}})(\Sb_{T_i}-\Sb_{T_i-})
$$
and
\begin{eqnarray*}
\int_{T_{i-1}}^{T_i} \hat\gamma_u(\hat \Pi(\Sb))d\hat \Pi_u(\Sb)
&=&
\sum_{k=1}^{M_i-1} \hat\gamma_{\hat\tau^{(i)}_{k}(\Sb)}(\hat \Pi(\Sb))
\left(\pi^{(n+k)}(\Sb_{\tau^{(i)}_{k}(\Sb)})-
\pi^{(n+k-1)}(\Sb_{ \tau^{(i)}_{k-1}(\Sb)})\right)\\
&&+(\hat\gamma_{T_i}-\hat\gamma_{\hat\tau^{(i)}_{M_i-1}})(\pi^{(n)}(\Sb_{T_i})-\pi^{(n+M_i-1)}(\Sb_{T_i-})).
\end{eqnarray*}
Again, by using the fact that
the portfolio $\hat \gamma$ is bounded by ${n}$,
we obtain the following estimate,
\begin{eqnarray} 
\nonumber
&&\left|\int _{0}^T \gamma_u(\Sb) d\Sb_u
-\int_{0}^T \hat\gamma_u(\hat \Pi(\Sb))d\hat \Pi_u(\Sb)\right|
 \leq
\sum_{i=1}^N \left|\int_{T_{i-1}}^{T_i} \gamma_u(\Sb) d\Sb_u
-\int_{T_{i-1}}^{T_i} \hat\gamma_u(\hat \Pi(\Sb))d\hat \Pi_u(\Sb)\right|
\\ &&\hspace{50pt}\leq
\nonumber
2 \|\hat \gamma\|_\infty \left(2N \sqrt d 2^{-n}+
\sum_{i=1}^N\sum_{k=1}^{M_i}
\left|  \pi^{(n+k)}(\Sb_{\tau^{(i)}_{k}})-
\Sb_{\tau^{(i)}_{k}}\right|\right)
\\ &&\hspace{50pt}\leq
\nonumber
 4N\sqrt d n 2^{-n}+2n \sum_{i=1}^N\sum_{k=1}^{M_i}\sqrt d 2^{-n-k}
\\ \label{7.9}
&&\hspace{50pt}\leq
6N\sqrt d n 2^{-n}.
\end{eqnarray}
By applying similar arguments as in
Lemma \ref{l.lift} we observe that $\gamma$ is progressively measurable
and $\int_{0}^t \gamma_u(\Sb) d\Sb_u$ is uniformly bounded from below.
The following lemma ends our preparations towards the proof of Theorem \ref{thm7.1}.
\begin{lem}\label{l.estimate}
i. Let $p>1$ given by (\ref{2.rev}). Then,
$$
V(\|\Sb\|^p)<\infty.
$$
ii. Let $\epsilon>0$. Define the stopping times $\tau^{(\epsilon)}_0=0$ and
for $j>0$
$$\tau^{(\epsilon)}_j=T\wedge\min\{t>\tau^{(\epsilon)}_{j-1}:
t\in \{T_1,...,T_{N-1}\} \ \ \mbox{or} \ \ |\Pi_t(\Sb)-\Pi_{\tau^{(\epsilon)}_{j-1}}(\Sb)|\geq\epsilon\}.$$
Set
$M^{(\epsilon)}=\min\{k: \tau^{(\epsilon)}_k=T\}.$
Consider the random variable
$$X_{\epsilon}=\sqrt{\sum_{i=1}^{M^{(\epsilon)}} |\Pi_{\tau^{(\epsilon)}_{i}}(\Sb)-\Pi_{\tau^{(\epsilon)}_{i-1}}(\Sb)|^2}.$$
Then $$V(X_{\epsilon})<3 d V(||\Sb||^p).$$
\end{lem}
\begin{proof}
i. Fix $n \in \N$.
Let $\tau_k$ and $n$ be as in
subsection \ref{ss.app}.
We define a portfolio $(g,\gamma)$ as
follows. Set $\gamma_0=0$.
For $k=0,1,...,n-1$ and $t\in (\tau_k,\tau_{k+1}]$,  let
$$
\gamma_t(\Sb)=-\frac{ p^2}{(p-1)}\left(\max_{0\leq i \leq k}(\Sb^{(1)}_{\tau_i})^{p-1},...,
\max_{0\leq i \leq k}(\Sb^{(d)}_{\tau_i})^{p-1}\right),
$$
and
$$
g(x)=\left(\frac{p}{p-1}\right)^p\sum_{i=1}^d x^p_i-\frac{p d}{p-1},
\quad x \in \R_+^d.
$$
We use Proposition 2.1 in
\cite{ABPST} to conclude that for any
$k=0,1,...,n-1$ and $t\in (\tau_k,\tau_{k+1}]$,
$$
g(\Sb_t)+\int_{0}^t \gamma_u d\Sb_u
\geq
\max(|\Sb_t|^p,\max_{0\leq i\leq k}|\Sb_{\tau_i}|^p).
$$
Therefore,
$\phi^{(n)}:=(g,\gamma)$
is admissible.
Also at $t=T$,
$$
g(\Sb_T)+\int_{0}^T \gamma_u d\Sb_u
\ge
\max_{0\leq i\leq n}|\Sb_{\tau_i}|^p.
$$
In view of  the definitions of $\tau_k$'s,
for sufficiently large $n$,
$$
\max_{0\leq i\leq n}|\Sb_{\tau_i}|^p
\geq
\left(||\Sb||-\sqrt{d}2^{-n}\right)^p\geq
\frac{||\Sb||^p}{2^p}-1.
$$
Combining all the above, we arrive at
 $$
 V(||\Sb||^p)\leq 2^p(1+ \int g d\mu_N)<\infty.$$
ii.
Define the trading strategy
$\gamma_t=\sum_{i=1}^{M^{(\epsilon)}} \gamma_i \chi_{(\tau^{(\epsilon)}_{i-1},\tau^{(\epsilon)}_{i}]}(t)$
where $\gamma_i=(\gamma^{(1)}_i,...,\gamma^{(d)}_i)$ is given by
$$\gamma^{(k)}_i=\left(\frac{-\Pi_{\tau^{(\epsilon)}_{i-1}}(\Sb^{(k)})}
{\sqrt{\sum_{j=1}^{i-1} |\Pi_{\tau^{(\epsilon)}_{j}}(\Sb^{(k)})-\Pi_{\tau^{(\epsilon)}_{j-1}}(\Sb^{(k)})|^2+\max_{0\leq j\leq i-1}
\Pi^2_{\tau^{(\epsilon)}_{j}}(\Sb^{(k)})}}\right).$$
From Theorem
Theorem 1.2 in \cite{BS} it follows that for any $i$,
$$\int_{0}^{\tau^{(\epsilon)}_i}\gamma_u d\Sb_u+ 3 d \max_{0\leq j\leq i}
\Pi_{\tau^{(\epsilon)}_{j}}(\Sb)\geq \sqrt{\sum_{j=1}^{i} |\Pi_{\tau^{(\epsilon)}_{j}}(\Sb)-\Pi_{\tau^{(\epsilon)}_{j-1}}(\Sb)|^2}.
$$
This together with the fact that $|\gamma|\leq \sqrt d$
yields that
$\gamma$ is admissible trading strategy, and
$V(X_{\epsilon}-3 d||S||)\leq 0$. Thus from the linearity of the market and the fact that $||\Sb||\geq||\Sb_0||=\sqrt d$ we get
$$V(X_{\epsilon})\leq 3 d V(||S||)\leq 3d V(||S||^p)$$
and the result follows.
\end{proof}
\subsection{Proof of Theorem \ref{thm7.1}}
\begin{proof}
We start with the proof of the inequality
\begin{equation}\label{7.10}
V(G)\geq\sup_{\mathbb Q\in\mathbb{M}_{\mu_1,...,\mu_N}}
 \E_{\mathbb {Q}}\left[G(\mathbb
S)\right].
\end{equation}
Let $\Q \in \M_{\mu_1,...,\mu_N}$.
Consider a trading strategy of the form
$\gamma=\gamma^{(1)}+\sum_{i=1}^{N-1} \beta_i\chi_{\{T_i\}}(t)$
where $\gamma^{(1)}$
is a portfolio trading strategy
 satisfies the same assumptions as
in Definition \ref{d.admissible} and $\beta_i$ is
$\mathcal{F}_{T_i-}$ measurable and bounded.
Clearly,
$\E_\Q(\Sb_{T_i}|\mathcal{F}_{T_i-})=\Sb_{T_i-}$, and so
$$\E_\Q\left[\sum_{i=1}^{N-1} \beta_i(\Sb_{T_i}-\Sb_{T_i-})\right]=0.$$
Thus,
$$
\E_\Q\left[ \int_0^T \gamma_u(\Sb)\ d\Sb_u\right]=
\E_\Q\left[ \int_0^T \gamma^{(1)}_u(\Sb)\ d\Sb_u\right]
\le 0.$$
Now suppose that $(g_1,...,g_N,\gamma)$
is an admissible super-replicating
semi-static portfolio.
Then,
$$\sum_{i=1}^N \int g_i d\mu_i=\E_\Q \left[\sum_{i=1}^N g_i(\Sb_{T_i}\right]\geq \E_\Q [G(\Sb)],$$
and we conclude (\ref{7.10}).

Next, prove the inequality
\begin{equation}\label{7.4}
V(G)\leq\sup_{\mathbb Q\in\mathbb{M}_{\mu_1,...,\mu_N}}
 \E_{\mathbb {Q}}\left[G(\mathbb
S)\right].
\end{equation}
The proof will be done in four steps.\\
${}$\\

\textbf{Step 1:}
In this step we show that if (\ref{7.4}) holds for a
bounded non negative $G$,
then it holds for
general function
satisfying Assumption \ref{a.newG}.
A similar reduction is already done in the 
proof of Theorem \ref{t.main}.  However,
that proof uses the growth assumption
\reff{1.new} while we now assume a weaker condition
\reff{e.growth}.  The proof below is essentially
the same as the one given in  our earlier papers
\cite{DS,DS1}.

First, assume that
$G$ is a claim satisfying
Assumption \ref{a.newG} that is also bounded
from below.
Thus there exists $M>0$ such that $G\geq -M$.
For $K>0$ large, set
$$
G_K:= G \wedge c(K+1)+M.
$$
Then, $G_K$ is bounded non negative, and so (\ref{7.4})
applies to $G_K$ yielding,
$$
V(G_K)\leq\sup_{{\mathbb Q}\in \mathbb{M}_{\mu_1,...,\mu_N}}
 \E_\Q \left[G_K(\Sb)\right] \le \sup_{{\mathbb Q}\in
\mathbb{M}_{\mu_1,...,\mu_N}} \E_\Q \left[G(\Sb)\right]+M.
$$
Moreover, by the upper bound on $G$,
the set $\{G(\Sb) \ge c(K+1)\}$ is
included in the set $\{ \|\Sb\| \geq K \}$. Hence,
\begin{eqnarray*}
G(\Sb)& \le & G_K(\Sb) + c \left(\|\Sb\|+1\right) \chi_{\{\|\Sb\| \geq K\}}(\Sb)-M\\
&\leq & G_K(\Sb)+ c \frac{(||\Sb||+1)^p}{K^{p-1}}-M\\
&\leq & G_K(\Sb)+  \frac{c 2^p}{K^{p-1}} \ ||\Sb||^p-M.
\end{eqnarray*}
By the linearity of the
market, this inequality implies that
$$
V(G) \le V(G_K) +  \frac{c 2^p}{K^{p-1}} V(||\Sb||^p)-M.
$$
Thus, for any $K>0$,
$$
 V(G)\leq
 \sup_{\Q\in \M_{\mu_1,...,\mu_N}} \E_\Q \left[G(\Sb)\right]
+\frac{c 2^p}{K^{p-1}} V(||\Sb||^p).
$$
We let $K$ tend to infinity and
apply Lemma \ref{l.estimate} to conclude
duality holds for all $G$
satisfying the Assumption \ref{a.newG}
and bounded from below.\\

Now suppose that $G$ is a general function
satisfying Assumption \ref{a.newG}.  For $K>0$ large, set
$$
\check{G}_K:= G \vee (-c[K+1]).
$$ Then, $\check{G}_K$ is
bounded from below and duality holds.
Again, the linear upper bound
implies that
$\check{G}_K\left( \Sb\right)\le G\left(\Sb\right) + \check e_K(\Sb)$, where the
error function is
$$
\check e_K(\Sb):=c \left(\| \Sb\| +1\right)\chi_{\{\| \Sb \| \ge K\}}(\Sb)
\le  \frac{c 2^p}{K^{p-1}} \ ||\Sb||^p.
$$
Since $G \le \check{G}_K$ and duality holds for $ \check{G}_K$,
\begin{eqnarray*}
V(G) &\le& V(\check{G}_K) = \sup_{\Q \in \M_{\mu_1,...,\mu_N}} \E_\Q[\check{G}_K]
\le \sup_{\Q \in\M_{\mu_1,...,\mu_N}} \E_\Q[G+\check{e}_K]\\
&\le &\sup_{\Q \in \M_{\mu_1,...,\mu_N}} \E_\Q[G]
+ \sup_{\Q \in \M_{\mu_1,...,\mu_N}} \E_\Q[\check{e}_K].
\end{eqnarray*}
Moreover, using the Doob's inequality
for the $\Q \in \M_{\mu_1,...,\mu_N}$ martingale $\Sb$, we obtain,
\begin{eqnarray*}
\sup_{\Q \in \M_{\mu_1,...,\mu_N}}\
\E_\Q \left[ \check e_K\left(\Sb\right)\right] & \leq &
  \frac{c 2^p}{K^{p-1}}  \sup_{\Q \in \M_{\mu_1,...,\mu_N}}\
  \E_\Q\left(||\mathbb S||^p \right)\\
&\leq &
C_p   \frac{c 2^p}{K^{p-1}}
\sup_{\Q \in \M_{\mu_1,...,\mu_N}}\ \E_\Q\left(\left|\Sb_T\right|^p\right)\\
&=&
C_p   \frac{c 2^p}{K^{p-1}}
\int |x|^p d\mu_N(x),
\end{eqnarray*}
where $C_p$ is the constant in the Doob's inequality.
Once again, we let $K$ tend to infinity to arrive at
(\ref{7.4}).

\textbf{Step 2:}
From know on, we assume that
$0\leq G\leq c$ for some $c>0$.
Fix $\epsilon>0$ and $n\in\mathbb N$.
On the space $\hat{\mathbb D}$
define the stopping times $\hat\tau^{(\epsilon)}_0=0$ and
for $j>0$
$$
\hat\tau^{(\epsilon)}_j=T\wedge\min\{t>\hat\tau^{(\epsilon)}_{j-1}:
t\in \{T_1,...,T_{N-1}\} \ \ \mbox{or} \ \ |\hat\Sb_t-\hat\Sb_{\hat\tau^{(\epsilon)}_{j-1}}|\geq\epsilon\}.
$$
Set
$\hat M^{(\epsilon)}=\min\{k: \hat\tau^{(\epsilon)}_k=T\}.$
Introduce the random variable
$$\hat X_{\epsilon}:= F(\hat \Sb):=\sqrt{\sum_{i=1}^{\hat M^{(\epsilon)}}
|\hat \Sb_{\hat\tau^{(\epsilon)}_{i}}-\hat\Sb_{\tau^{(\epsilon)}_{i-1}}|^2}$$
and consider the bounded claim
$$Y=G(\hat\Sb)-\left(\frac{c}{\epsilon}\wedge \epsilon \hat X_{\epsilon}\right).$$
Define the set
$\cM(n,c)$
of all probability measures
which satisfy
\begin{equation}\label{7.ne}
\sum_{m \in \N^d}\ \left|\Q\left(\hat\Sb_{T_i}=m 2^{-n} \right) -\hat \mu_i
\left(\left\{m  2^{-n} \right\}\right)\right|\leq \frac{c}{n}, \ \ i=1,...,N
\end{equation}
and (\ref{e.Qn2}).
Also let
$\cM(n,c,\epsilon)\subset \cM(n,c)$ be the set of all probability measures
$\Q$ which in addition satisfy
$\E_\Q \left[\frac{c}{\epsilon}\wedge \epsilon \hat X_{\epsilon}\right]\leq c$.
From the Markov inequality it follows that for any
$\Q\in \cM(n,c,\epsilon)$
\begin{equation}\label{7.11}
\mathbb Q\left(\hat X_{\epsilon}\geq \frac{c}{\epsilon^2}\right)\leq\epsilon.
\end{equation}
Using similar arguments as in Lemma \ref{l.upperest} it follows that
\begin{equation}\label{7.12}
{V}^{(n)}(Y) \leq \left[\sup_{\mathbb Q\in
\cM(c,n)} \E_{\mathbb Q} Y\right]^{+}=
\left[\sup_{\mathbb Q\in
\cM(c,n,\epsilon)} \E_{\mathbb Q} Y\right]^{+},
\end{equation}
where the last equality follows from the fact that
$G\leq c$.

Next, from the linearity of the market and Lemma
\ref{l.estimate} we have
\begin{equation}\label{7.13}
V(G)\leq V\left(G-\frac{c}{\epsilon}\wedge\epsilon X_{\epsilon}\right)+\epsilon V(X_{\epsilon})\leq
V\left(G-\frac{c}{\epsilon}\wedge\epsilon  X_{\epsilon}\right)
+c_2\epsilon
\end{equation}
for some constant $c_2$.

Finally, we estimate the term
$V\left(G-\frac{c}{\epsilon}\wedge\epsilon X_{\epsilon}\right)-V^{(n)}(Y)$,
from above.

From Assumption \ref{a.newG}, (\ref{7.5})--(\ref{7.8}) and the fact $0\leq G\leq c$ we obtain that for $n$
sufficiently large,
\begin{equation*}
|G(\Sb)-G(\hat\Pi(\Sb))|\leq \epsilon+c \chi_{||S||\geq \epsilon^{-1}}\leq \epsilon+c\epsilon^{p-1}||S||^p.
\end{equation*}
Observe that
$X_{\epsilon}=F(\hat\Pi(\Sb))$. Thus
from (\ref{7.9}),
Lemma \ref{l.estimate} and the linearity of the market we get
\begin{equation}\label{7.14}
V\left(G-\frac{c}{\epsilon}\wedge\epsilon X_{\epsilon}\right)-V^{(n)}(Y)\leq 6N\sqrt d n 2^{-n}+\epsilon+ \epsilon^{p-1} V(||S||^p)\leq c_3 \epsilon^{p-1}
\end{equation}
for some constant $c_3$.
From (\ref{7.12})--(\ref{7.14}) it follows
that for $n$ sufficiently large,
\begin{equation}\label{7.15}
V(G)\leq c_4 \epsilon^{p-1}+ \left[\sup_{\mathbb Q\in
\cM(c,n,\epsilon)} \E_{\mathbb Q} [G(\hat \Sb)]\right]^{+}
\end{equation}
for some constant $c_4$.
\\
\textbf{Step 3:}
In order to complete the proof of the theorem it remains
to establish that
\begin{equation}\label{7.101}
\lim\sup_{n\rightarrow\infty}\left[\sup_{\mathbb Q\in
\cM(c,n,\epsilon)} \E_{\mathbb Q} [G(\hat \Sb)]\right]^{+}\leq\sup_{\mathbb Q\in\mathbb{M}_{\mu_1,...,\mu_N}}
 \E_{\mathbb {Q}}\left[G(\mathbb
S)\right]+m(\epsilon)
\end{equation}
where $m:\mathbb{R}_{+}\rightarrow\mathbb{R}_{+}$ is a continuous function
with $m(0)=0$. Then
by letting $\epsilon\downarrow 0$ we obtain the duality.

Clearly, we can assume that for $n$ sufficiently large
the set $\cM(c,n,\epsilon)\neq\emptyset$ is not empty,
otherwise the left hand side of (\ref{7.101})$=0$ and the statement is trivial.

We start with a modification of the process $\hat\Sb$.
Namely, we will modify the stochastic process $\hat S$, such that the new process
will have a finitely many (uniformly bounded) jumps. This modification will allow us to obtain
tightness.

Thus fix $n\in\mathbb N$ (sufficiently large). There
exists a probability measure $\Q_n\in\cM(c,n,\epsilon)$ such that
\begin{equation}\label{7.101+}
\E_{\mathbb Q_n} [G(\hat \Sb)]>\left[\sup_{\mathbb Q\in
\cM(c,n,\epsilon)} \E_{\mathbb Q} [G(\hat \Sb)]\right]^{+}-1/n.
\end{equation}
Define the process
$$\tilde\Sb_t=\sum_{i=1}^N \hat\Sb_{T_{i-1}+\alpha_i(t-T_{i-1})}\chi_{[T_{i-1},T_i-\epsilon)}(t)+
\Sb_{T_i}\chi_{[T_i-\epsilon,T_i]}(t)$$
where
$\alpha_i=\frac{T_i-T_{i-1}}{T_{i}-T_{i-1}-\epsilon}$, $i=1,...,N$.
Observe that the Skorokhod distance between $\tilde\Sb$ and $\hat\Sb$ satisfies
${d}(\tilde\Sb,\hat\Sb)\leq\epsilon$
and
$$\left|\int_{0}^T\hat\Sb_u du-\int_{0}^T\tilde\Sb_u du\right|\leq 2N\epsilon||S||.$$
Thus $\check d(\tilde\Sb,\hat\Sb)\leq (2N+1)\epsilon||S||.$ This together
with Assumption \ref{a.newG} and the fact that $0\leq G\leq c$ yields
\begin{equation}\label{7.102}
|G(\tilde\Sb)-G(\hat\Sb)|\leq m_G((4N+2)d\sqrt\epsilon)+c\chi_{||\hat\Sb||\geq 2d\epsilon^{-1/2}}.
\end{equation}
Similarly to Lemma \ref{l.upperest}
we have the decomposition
$\hat\Sb=\M^{\Q_n}-A^{\Q_n}$. Denote
$\M^{\Q_n}=(M^{(1)},...,M^{(d)})$
and $A^{\Q_n}=(A^{(1)},...,A^{(d)})$. Observe that (since $\Q_n\in\cM(c,n,\epsilon)$)
for any $i$, $\E_{\Q_n} M^{(i)}_T=1$ and $\E_{\Q_n} ||A^{(i)}||\leq \frac{c}{n}$.
Thus
from
the Doob inequality
and the Markov inequality we obtain
\begin{equation}\label{7.412}
\Q_n(||\hat\Sb||\geq 2d\epsilon^{-1/2})\leq
\sum_{i=1}^d [\Q_n(||M^{(i)}||\geq\epsilon^{-1/2})+
\Q_n(||A^{(i)}||\geq\epsilon^{-1/2})]\leq
d\sqrt\epsilon(1+c/n).
\end{equation}
This together with (\ref{7.102}) gives
\begin{equation}\label{7.103}
|\E_{\Q_n} [G(\tilde\Sb)]-\E_{\Q_n}[G(\hat\Sb)]|\leq m_G(c_5\sqrt\epsilon)+c_5\sqrt\epsilon
\end{equation}
for some constant $c_5$.

Next, set $\Theta= \lceil N+c^2/\epsilon^6\rceil$ and
$\delta=\frac{\epsilon}{4 \Theta^2}$.
Define
$\tilde\tau_0=0$ and for
$1\leq j\leq\Theta $
define
$$\tilde\tau_j=
(T-\delta)\wedge\min\{t>\tilde\tau_{j-1}:
t\in \{T_1,...,T_{N-1}\} \ \ \mbox{or} \ \ |\tilde\Sb_t-\tilde\Sb_{\tilde\tau_{j-1}}|\geq\epsilon\}.$$
For $j>\Theta$
we set
$\tilde\tau_j=(T-\delta)\wedge\min\{T_i: T_i>\tilde\tau_{j-1}\}$. Observe that
$\tilde\tau_{N+\Theta}=T-\delta$.

Let $\sigma_0=0$ and for $k>0$ let
$\sigma_k=\tilde\tau_k+\delta k $ if
$\tilde\tau_k\not\in\{T_1,...,T_{N-1},T-\delta\}$
and $\sigma_k=\tilde\tau_k$ otherwise.
Define the process
$$\check\Sb_t=\sum_{i=0}^{\Theta+N-1}\tilde\Sb_{\tilde\tau_i} \chi_{[\sigma_i,\sigma_{i+1})}(t)+
\hat\Sb_T\chi_{[T-\delta,T]}(t).$$
Recall the inequality (\ref{7.11}).
Observe that on the event $\{\hat X_{\epsilon}\geq \frac{c}{\epsilon^2}\}$
we have
$$\min\{k:\tilde\tau_k=T-\delta\}\leq \Theta.$$
Thus
\begin{equation}\label{7.202-}
\mathbb Q_n(\tilde\sigma_{\Theta}=T-\delta)\geq 1-\epsilon
\end{equation}
and so,
$$
d(\check\Sb,\tilde\Sb)\leq 2\epsilon+\max_{1\leq i\leq M+\Theta} [\sigma_i-\tau_i]\leq 3\epsilon.
$$
Furthermore, similarly to (\ref{7.8}) we get
$$
\left|\int_{0}^T\check\Sb_t dt-\int_{0}^T\tilde\Sb_t dt\right|\leq 2\epsilon T+2\epsilon \|\hat\Sb\|
\leq
c_6\epsilon \|\hat\Sb\|
$$
for some constant $c_6$.
This observation together with
(\ref{7.11}) yields that
\begin{eqnarray}\label{7.202}
\nonumber
\left|\E_{\Q_n}[G(\tilde \Sb)]-\E_{\Q_n}[G(\check\Sb)]\right|
&\leq & c\epsilon+c \Q_n(||\hat\Sb||\geq 2d\epsilon^{-1/2})+
m_G(3\epsilon+2d c_6\sqrt\epsilon) \\
& \leq & c\epsilon+ c d\sqrt\epsilon(1+c/n)+
m_G(3\epsilon+2d c_6\sqrt\epsilon).
\end{eqnarray}
From (\ref{7.101+}),
(\ref{7.103})
and (\ref{7.202}) it follows that in order to establish
(\ref{7.101}) it sufficient to show
\begin{equation}\label{7.300}
\lim\sup_{n\rightarrow\infty}\E_{\Q_n}[G(\check\Sb)]\leq \sup_{\mathbb Q\in\mathbb{M}_{\mu_1,...,\mu_N}}
 \E_{\mathbb {Q}}\left[G(\mathbb
S)\right]+m_G(\epsilon)+c \epsilon.
\end{equation}
\textbf{Step 4:}
Finally, we establish (\ref{7.300}) by using weak convergence on the
Skorokhod space $\mathbb D$. Without loss of generality (by passing to a
subsequence)
we assume that the limit in the left hand side of (\ref{7.300})
is exists.

In this step we denote the
process $\tilde\Sb,\check\Sb$
and the stopping times $\tilde\tau_k,\sigma_k$,
which constructed for $n\in\mathbb N$ by
$\tilde\Sb^{(n)},\check\Sb^{(n)}$ and $\tilde\tau^{(k)}_n,\sigma^{(n)}_k$,
respectively.

Introduce the martingale
$$\tilde \M^{(n)}_t=\sum_{i=1}^N \M^{\Q_n}_{T_{i-1}+\alpha_i(t-T_{i-1})}\chi_{[T_{i-1},T_i-\epsilon)}(t)+
\M^{\Q_n}_{T_i}\chi_{[T_i-\epsilon,T_i]}(t).$$
For $k=0,1,...,N+\Theta$ let
$$X^{(n)}_k=\check\Sb^{(n)}_{\sigma^{(n)}_k}=\tilde\Sb^{(n)}_{\tilde\tau^{(n)}_k}, \ \
Y^{(n)}_k=\tilde\M^{(n)}_{\tilde\tau^{(n)}_k},  \ \
Z^{(n)}_k=\tilde\M^{(n)}_{\tilde\tau^{(n)}_k-} \ \ \mbox{and} \ \ W^{(n)}_k=\tilde\Sb^{(n)}_{\tilde\tau^{(n)}_k-}.
$$
From (\ref{7.ne}) it follows that we have a weak convergence
$\check\Sb^{(n)}_T\Rightarrow \mu_N$.
In addition, from the fact that $\check\Sb^{(n)}_T\geq 0$ and
$$\lim_{n\rightarrow\infty}\mathbb E_{\Q_n}[\check\Sb^{(n)}_T]=
\lim_{n\rightarrow\infty}\mathbb E_{\Q_n}[\M^{\mathbb Q_n}_T]=
(1,...,1)=\int x d\mu_N(x)$$
it follows that the sequence
${\{\check\Sb^{(n)}_T\}}_{n=1}^\infty$
is uniformly integrable. In addition, the equality
$\lim_{n\rightarrow\infty}\mathbb E_{\Q_n}||A^{\mathbb Q_n}||=0$
yields that
${\{\M^{\Q_n}_T\}}_{n=1}^\infty$
is uniformly integrable, and since $M^{\Q_n}$ is a martingale
we can replace $T$ by any stopping time.
Thus, we conclude that the
the sequence
$$
\left(X^{(n)}_0,...,X^{(n)}_{N+\Theta},
Y^{(n)}_0,...,Y^{(n)}_{N+\Theta},Z^{(n)}_0,...,Z^{(n)}_{N+\Theta},
\sigma^{(n)}_0,...,\sigma^{(n)}_{N+\Theta}\right),
\ \ n\in\mathbb N
$$
is uniformly integrable, and in particular its
tight on the space $\mathbb{R}^{4N+4\Theta+4}$. Thus there is a subsequence
(which we still denote by $n$) which converge weakly.
From the Skorokhod representation theorem it follows that we can
redefine the above sequence on a new probability space such that it converge a.s.
Denote the limit by
$$\left(X_0,...,X_{N+\Theta},Y_0,...,Y_{N+\Theta},Z_0,...,Z_{N+\Theta},\sigma_0,...,\sigma_{N+\Theta}\right)$$
and introduce the \cad processes
$$U_t=\sum_{i=0}^{N+\Theta-1} X_i\chi_{[\sigma_i,\sigma_{i+1})}(t)+X_{N+\Theta}\chi_{[T-\delta,T]}(t).$$
Observe that for any $k,n$
$\sigma^{(n)}_k-\sigma^{(n)}_{k-1}>\delta $
provided that $\sigma^{(n)}_{k-1}<T-\delta$.
Thus
we get the same property for the limit, namely
$\sigma_k-\sigma_{k-1}> \delta $
provided that $\sigma_{k-1}<T-\delta$.
We conclude
that $\check\Sb^{(n)}\rightarrow U$ a.s with respect to the Skorokhod topology
on the space $\D$.
Thus $G(\check\Sb^{(n)})\rightarrow G(U)$ a.s, and so from the bounded convergence
theorem it follows that
\begin{equation}\label{7.500}
\mathbb E [G(U)]=\lim_{n\rightarrow\infty}\E_{\Q_n}[G(\check\Sb)].
\end{equation}
Let us notice that $U$ is not a martingale, and so we modify $U$.

Let $\mathcal G^{(i)}_t$ be the right
continuous filtration which
given by
$$\mathcal G^{(i)}_t=
\bigcap_{u>t}\sigma\{Y_0,...,Y_i,\sigma_0,...,\sigma_i,u\wedge\sigma_{i+1}\}.
$$
Introduce the \cad process
$$\tilde U_t=\sum_{i=0}^{N+\Theta-1} 
\mathbb E(Z_{i+1}|\mathcal G^{(i)}_t)\chi_{[\sigma_i,\sigma_{i+1})}(t)
+X_{N+\Theta}\chi_{[T-\delta,T]}(t).
$$
From the fact that $\lim_{n\rightarrow\infty}\mathbb E_{\Q_n}||A^{\Q_n}||=0$
it follows that
\begin{eqnarray}\label{7.502}
&X_k=Y_k,  \ \ Z_k=W_k, \ \ k=0,1,...,N+\Theta,\\
&\mbox{and} \ \ |W_{k}-X_{k-1}|\leq \epsilon, \ \ k=0,1,...,\Theta.\nonumber
\end{eqnarray}
Next, observe that for a given $n$ we have
$$
\mathbb E_{\Q_n}(Z^{(n)}_{k+1}|\sigma^{(n)}_1,...,
\sigma^{(n)}_{k},Y^{(n)}_1,...,Y^{(n)}_{k})=Y^{(n)}_k.
$$
and
$$
\mathbb E_{\Q_n}(Y^{(n)}_{k+1}|\sigma^{(n)}_1,...,
\sigma^{(n)}_{k+1},Z^{(n)}_1,...,Z^{(n)}_{k+1},Y^{(n)}1,...,Y^{(n)}_k)=Z^{(n)}_{k+1}.
$$
This together with uniform integrability yields
\begin{equation}\label{7.503}
\mathbb E (Z_{k+1}|\sigma_1,...,\sigma_k,Y_1,...,Y_k)=Y_k
\end{equation}
and
\begin{equation}\label{7.504}
\mathbb E (Y_{k+1}|\sigma_1,...,\sigma_{k+1},Z_1,...,Z_{k+1},Y_1,...,Y_k)=Z_{k+1}.
\end{equation}
From (\ref{7.503})--(\ref{7.504}) and
the chain rule for conditional expectation it follows that
$\tilde U$ is a martingale.
From (\ref{7.502})--(\ref{7.503})
we have
$\tilde U_{\sigma_k}=Y_k=X_k=U_{\sigma _k}$.
Observe that if $\sigma_k=T_i$ for some $i$, then for sufficiently
large $n$ we have
$\sigma^{(n)}_i=T$.
Thus
$$U_{T_i}=\tilde U_{T_i}=\lim_{n\rightarrow\infty} \hat \Sb^{(n)}_{T_i}.$$
This together with (\ref{7.ne})
gives that for any $i$ the distribution of
$\tilde U_{T_i}$ is equals to $\mu_i$, we conclude that the law of $\tilde U$ is an element in
$\mathbb M_{\mu_1,...,\mu_N}$.

Finally, we estimate
$\mathbb E [G(U)]-\mathbb E [G(\tilde U)]$.
Let $k<\Theta$. On the event $t\in [\sigma_k,\sigma_{k+1})$ (which is $\mathcal G^{(i)}_t$ measurable)
we apply (\ref{7.502})--(\ref{7.503}) to obtain
$$
|\tilde U_t- U_t|=|\mathbb E(Z_{k+1}-Y_k|Y_1,...,Y_k,\sigma_1,...,\sigma_k)|\leq \epsilon.
$$
Thus on the event
$\sigma_{\Theta}=T-\delta$ we get $||U-\tilde U||\leq \epsilon$.
We conclude that
\begin{equation}\label{7.final}
|\mathbb E G(U)-\mathbb EG(\tilde U)|
\leq c\mathbb P(\sigma_{\Theta}<T-\delta)+m_G(\epsilon)\leq c\epsilon+m_G(\epsilon)
\end{equation}
where the last inequality follows from
(\ref{7.202-}).
By combining
(\ref{7.500}) with (\ref{7.final})
we obtain (\ref{7.300}), and complete the proof.
\end{proof}

\end{document}